\documentclass[11pt]{article}

\usepackage[a4paper, margin=0.75in]{geometry}
\usepackage{fourier} 
\usepackage{booktabs} 
\usepackage[T1]{fontenc} 
\usepackage[utf8]{inputenc}
\usepackage{authblk} 
\usepackage{footmisc}

\usepackage[table,dvipsnames]{xcolor} 
\definecolor{orcidgreen}{HTML}{85A12C} 
\usepackage{url}
\usepackage[hidelinks]{hyperref} 
\hypersetup{
    colorlinks=true,
    allcolors=orcidgreen,
    linkcolor=Blue,
    citecolor=purple,
}
\usepackage{microtype} 

\usepackage{pdflscape}
\usepackage{multirow}
\usepackage{breakcites}
\usepackage{makecell}
\usepackage{subcaption}
\usepackage{wrapfig}
\usepackage{longtable}

\newcommand{\boldpara}[1]{\medskip\noindent\textbf{#1.}}

\usepackage{fontawesome5}
\usepackage{pifont}
\definecolor{crossred}{RGB}{201, 22, 22}
\definecolor{tickgreen}{RGB}{38, 150, 68}
\newcommand{\yes}{\textcolor{tickgreen}{\ding{51}}}
\newcommand{\no}{\textcolor{crossred}{\ding{55}}}

\usepackage{amssymb}
\usepackage{amsmath}
\usepackage{amsthm}
\usepackage{orcidlink}
\usepackage{academicons}
\usepackage{xspace}
\usepackage{nicefrac}

\newcommand{\ham}{{\ensuremath{\mathrm{ham}}}}
\newcommand{\jacc}{{\ensuremath{\mathrm{jacc}}}}

\newcommand{\ehd}{{\ensuremath{\mathrm{ehd}}}}

\newcommand{\EHD}{{\ensuremath{\mathrm{EHD}}}}
\newcommand{\cen}{{\ensuremath{\mathrm{cen}}}}
\newcommand{\chd}{{\ensuremath{\mathrm{chd}}}}
\newcommand{\CHD}{{\ensuremath{\mathrm{CHD}}}}
\newcommand{\avl}{{\ensuremath{\mathrm{avl}}}}
\newcommand{\revavl}{{\ensuremath{\mathrm{rev\hbox{-}avl}}}}

\newcommand{\satr}{{\ensuremath{\mathrm{satr}}}}
\newcommand{\outdiv}{{\ensuremath{\mathrm{out\hbox{-}div}}}}

\newcommand{\noisy}{{\mathrm{N}}}
\newcommand{\calD}{{{\mathcal{D}}}}
\newcommand{\bernoulli}{\ensuremath{\mathrm{Bernoulli}}}

\newcommand{\pcc}{{\mathrm{pcc}}}

\newcommand{\agr}{a}
\newcommand{\std}{\mathrm{std}}
\newcommand{\cmpl}{\mathrm{cmpl}}
\newcommand{\cntragr}{{\mathrm{cntr\hbox{-}agr}}}
\newcommand{\pairagr}{{\mathrm{pair\hbox{-}agr}}}
\newcommand{\jaccpairagr}{{\mathrm{jacc\hbox{-}agr}}}
\newcommand{\appragr}{{\mathrm{av\hbox{-}agr}}}
\newcommand{\pccagr}{{\mathrm{pcc\hbox{-}agr}}}
\newcommand{\pospccagr}{{\mathrm{pcc^+\hbox{-}agr}}}
\newcommand{\cntrdiv}{{\mathrm{cntr\hbox{-}div}}}
\newcommand{\pccdiv}{{\mathrm{pcc\hbox{-}div}}}
\newcommand{\cntrpol}{{\mathrm{cntr\hbox{-}pol}}}
\newcommand{\pccpol}{{\mathrm{pcc\hbox{-}pol}}}
\newcommand{\pairpol}{{\mathrm{pair\hbox{-}pol}}}

\newcommand{\ID}{\mathrm{ID}}
\newcommand{\IC}{\mathrm{IC}}

\newcommand{\reals}{{\mathbb{R}}}

\usepackage[nameinlink]{cleveref} 

\crefname{claim}{Claim}{Claims}

\usepackage[square]{natbib}
\theoremstyle{plain}
\newtheorem{theorem}{Theorem}[section]
\newtheorem{lemma}[theorem]{Lemma}
\newtheorem{proposition}[theorem]{Proposition}

\newtheorem{definition}{Definition}[section]

\newtheorem{remark}{Remark}[section]

\definecolor{GreyishGrey}{RGB}{114,114,114}
\definecolor{BluishBlue}{RGB}{57,114,167}
\definecolor{GreenishGreen}{RGB}{57,167,114}

\newcommand{\ltbase}[1]{%
  \raisebox{3mm}{\makecell{%
    \edef\smallvalue{\fpeval{round(100*#1/1.75,2)}}%
    \begingroup
      \setlength{\fboxsep}{0pt}%
      \colorbox{GreyishGrey!\smallvalue}{\makebox[5ex][c]{\rule{0pt}{2ex}#1}}%
    \endgroup
  }}%
}
\newcommand{\ltagr}[1]{%
  \raisebox{3mm}{\makecell{%
    \edef\smallvalue{\fpeval{round(100*#1/1.75,2)}}%
    \begingroup
      \setlength{\fboxsep}{0pt}%
      \colorbox{GreenishGreen!\smallvalue}{\makebox[5ex][c]{\rule{0pt}{2ex}#1}}%
    \endgroup
  }}%
}
\newcommand{\ltdiv}[1]{%
  \raisebox{3mm}{\makecell{%
    \edef\smallvalue{\fpeval{round(100*#1/1.75,2)}}%
    \begingroup
      \setlength{\fboxsep}{0pt}%
      \colorbox{BluishBlue!\smallvalue}{\makebox[5ex][c]{\rule{0pt}{2ex}#1}}%
    \endgroup
  }}%
}
\newcommand{\ltpol}[1]{%
  \raisebox{3mm}{\makecell{%
    \edef\smallvalue{\fpeval{round(100*#1/1.75,2)}}%
    \begingroup
      \setlength{\fboxsep}{0pt}%
      \colorbox{red!\smallvalue}{\makebox[5ex][c]{\rule{0pt}{2ex}#1}}%
    \endgroup
  }}%
}
\newcommand{\qqagr}[1]{%
    \pgfmathsetmacro\smallvalue{(100*#1)/1.75}
    \tikz[baseline=0.3ex]{%
        \fill[GreenishGreen!\smallvalue!white!] (-0.1,0ex) rectangle (6.2ex,2ex);
        \node[inner sep=0pt, anchor=center] at (3ex,1ex) {#1};%
    }%
}
\newcommand{\qqdiv}[1]{%
    \pgfmathsetmacro\smallvalue{(100*#1)/1.75}
    \tikz[baseline=0.3ex]{%
        \fill[BluishBlue!\smallvalue!white!] (-0.1,0ex) rectangle (6.2ex,2ex);
        \node[inner sep=0pt, anchor=center] at (3ex,1ex) {#1};%
    }%
}
\newcommand{\qqpol}[1]{%
    \pgfmathsetmacro\smallvalue{(100*#1)/1.75}
    \tikz[baseline=0.3ex]{%
        \fill[red!\smallvalue!white!] (-0.1,0ex) rectangle (6.2ex,2ex);
        \node[inner sep=0pt, anchor=center] at (3ex,1ex) {#1};%
    }%
}
\newcommand{\qqpcor}[1]{%
    \pgfmathsetmacro\smallvalue{(100*#1)/1.75}
    \tikz[baseline=0.3ex]{%
        \fill[GreenishGreen!\smallvalue!white!] (-0.1,0ex) rectangle (6.2ex,2ex);
        \node[inner sep=0pt, anchor=center] at (3ex,1ex) {#1};%
    }%
}
\newcommand{\qqncor}[1]{%
    \pgfmathsetmacro\smallvalue{(100*#1)/1.75}
    \tikz[baseline=0.3ex]{%
        \fill[red!\smallvalue!white!] (-0.1,0ex) rectangle (6.2ex,2ex);
        \node[inner sep=0pt, anchor=center] at (3ex,1ex) {-#1};%
    }%
}

\apptocmd{\sloppy}{\hbadness 10000\relax}{}{}
\apptocmd{\appendix}{\hbadness 10000\relax}{}{}

\title{\huge \textbf{Agreement, Diversity, and Polarization Indices \\ for Approval Elections}}

\author{
  Piotr Faliszewski\quad
  Jitka Mertlová\quad
  Krzysztof Sornat\quad
  Stanisław Szufa\quad
  Tomasz Wąs
}

\date{}

\makeatletter
\renewcommand*{\@fnsymbol}[1]{\ifcase#1\or i\or ii\or iii\or iv\else\@ctrerr\fi}
\makeatother

\begin{document}

\maketitle

\begin{quote}
    \textbf{Abstract:}
    An index is a function that given an election outputs a value
  between $0$ and $1$, indicating the extent to which this election
  has a particular feature. We seek indices that capture agreement,
  diversity, and polarization among voters in approval elections, and
  that are normalized with respect to
  saturation. By the latter we mean that if two elections differ by
  the fraction of candidates approved by an average voter, but
  otherwise are of similar nature, then they should have similar index
  values. We propose several indices, analyze their properties, and
  use them to (a) derive a new map of approval elections, and (b) show
  similarities and differences between various real-life elections
  from Pabulib, Preflib and other sources.
\end{quote}
\begin{quote}
    \textbf{Code:} \url{https://github.com/Project-PRAGMA/Approval-DAP-IJCAI-2026}
\end{quote}

\begingroup
\renewcommand{\thefootnote}{}%
\renewcommand{\theHfootnote}{authorinfo}%
\footnotemark
\footnotetext{%
  \hspace{-20pt}
  Authors' Information:
  \href{https://orcid.org/0000-0002-0332-4364}{Piotr Faliszewski \orcidlink{0000-0002-0332-4364}}, faliszew@agh.edu.pl, AGH University of Kraków, Poland;
  \href{https://orcid.org/0009-0002-7778-0677}{Jitka Mertlová~\orcidlink{0009-0002-7778-0677}}, mertlova.jita@gmail.com, Czech Technical University in Prague, Czech Republic;
  \href{https://orcid.org/0000-0001-7450-4269}{Krzysztof~Sornat~\orcidlink{0000-0001-7450-4269}}, sornat@agh.edu.pl, AGH University of Kraków, Poland;
  \href{https://orcid.org/0000-0001-6301-6227}{Stanisław Szufa \orcidlink{0000-0001-6301-6227}}, s.szufa@gmail.com, University of Geneva, Switzerland;
  \href{https://orcid.org/0000-0003-3492-6584}{Tomasz Wąs \orcidlink{0000-0003-3492-6584}}, tomasz.was@cs.ox.ac.uk, University of Oxford, United Kingdom.
}%
\addtocounter{footnote}{-1}%
\endgroup


\section{Introduction}\label{sec:intro}

Our main goal is to design three functions that given an approval
election---i.e., an election where voters give yes/no answers
regarding their support for particular candidates---indicate the
extent to which the votes 
are in agreement, are diverse, or are polarized. Following
\citet{has-end:c:diversity-indices}, we refer to functions that
quantify election features (by giving values between $0$ and $1$) as
\emph{election indices}.

Such indices would be valuable because they would allow us to
compare approval elections coming from different sources. For example,
it is not uncommon to use participatory budgeting (PB) elections from
Pabulib~\citep{fal-fli-pet-pie-sko-sto-szu-tal:c:pabulib} outside of
the PB context and it is important to validate if this data remains
relevant.\footnote{E.g., for year 2025 Google Scholar points to 18
  papers that cite the Pabulib database and perform numerical
  experiments. Among these, only 12 are indeed focused on
  participatory budgeting.}  Further, as recently shown by
\citet{fal-mer-nun-szu-was:c:diff-sizes}, for ordinal elections (where the voters
rank the candidates) having good agreement,
diversity, and polarization indices leads to appealing maps of
elections
(i.e., visualizations of election datasets where each
election is a dot and distances between these dots give an idea of
their elections' similarity).  While a map of approval elections
exists~\citep{szu-fal-jan-lac-sli-sor-tal:c:sampling-approval-elections},
it seems to be less useful than those for ordinal ones (see the
argument below).

One characteristic issue when dealing with approval elections
is that they may largely differ in terms of their \emph{saturation},
i.e., the average number of candidates approved by a voter, while
being very similar in nature. As an extreme example, consider one
election where all voters approve the same 10\% of candidates, and a
second one, where all voters approve the same 50\% of
candidates. Both of these elections represent perfect agreement (and
complete lack of diversity and polarization) but according to the map
of \citet{szu-fal-jan-lac-sli-sor-tal:c:sampling-approval-elections},
they are quite far from each other. Indeed,
that map has two main dimensions, election saturation and the position
on the agreement-disagreement spectrum.  We seek
saturation-independent agreement, diversity, and polarization indices,
so that a map based on them would be capable of discovering more
nuanced relations between data.

Our general strategy
is to focus on designing a saturation-independent agreement index.
Given such an index, we derive
a diversity one
by clustering the voters into as like-minded groups as possible and claiming
that the election is diverse if the agreement within these groups is
low. For polarization--which we define as the presence of two groups
with conflicting preferences---we take the difference between the
overall agreement and the agreements in two voter groups after
clustering (the overall agreement in a polarized election would be
low, but it would be high in the like-minded groups).  This approach
was proposed by \citet{fal-kac-sor-szu-was:c:div-agr-pol-map} in the
ordinal case,
for the special case of measuring agreement
using the Kemeny rule~\cite{kem:j:no-numbers}. 
We extend their idea.

Altogether, our contributions are as follows. First, we propose a
number of agreement indices and analyze them theoretically and
experimentally, paying careful attention to saturation independence.
In particular, we verify if their results are intuitive on a number of
selected election instances. Then, using these indices, we derive and
test diversity and polarization ones (using both the general scheme
and some other ideas). Finally, we identify three indices---capturing
agreement, diversity, and polarization---that work best together in a
certain measurable way, and use them to form a map of approval
elections. Our map contains both synthetic and real-life elections,
e.g., coming from
Preflib~\citep{mat-wal:c:preflib} and
Pabulib~\citep{fal-fli-pet-pie-sko-sto-szu-tal:c:pabulib}, 
but also from other sources,
including some new ones.

We discuss related work throughout the paper, in relevant
places. 
Omitted proofs and discussions are in the appendix.

\section{Preliminaries}\label{sec:prelim}

For an integer $t$, we write $[t]$ to denote $\{1, \ldots, t\}$.  For
a vector $x \in \reals^t$, we write $x[i]$ to denote its $i$-th
coordinate, and $\overline{x} = \frac{1}{t}(x[1] + \ldots + x[t])$ to
denote its mean.  For $p \in [0,1]$, $x \sim \bernoulli(p)$ indicates
that $x$ gets value $1$ with probability $p$ and value $0$ with
probability $1-p$.

\paragraph{Elections.}
An \emph{(approval) election} is a pair $E = (C,V)$, where
$C = \{c_1, \ldots, c_m\}$ is a set of \emph{candidates} and
$V = (v_1, \ldots, v_n)$ is a collection of \emph{voters}.  Each voter
has an \emph{approval ballot} (or, a \emph{vote}) where he or she
indicates which candidates he or she approves
(and he or she disapproves the other ones). 
Formally, we represent a vote as a binary vector, whose
entries correspond to the candidates and have value $1$ if a given
candidate is approved and $0$ otherwise. To simplify notation, we
write $v_i$ to refer both to the respective voter and his or her vote;
the exact meaning will always be clear from the context.
Additionally, for a voter $v_i$ we write $A(v_i)$ to denote the set of
candidates that this voter approves, and for a candidate $c_j$ we
write $A(c_j)$ to mean the set of voters approving $c_j$.
Given an election
$E = (C,V)$, by $E^{-1}$ we mean an election obtained from $E$ by
reversing all entries of all votes (i.e., for each voter $v_i$ and
candidate $c_j$, in $E^{-1}$, for $c_j$ this voter reports
$1-v_i[j]$). We refer to $E^{-1}$ as a reverse of $E$.

\paragraph{Pearson Correlation Coefficient.}
Given two $t$-di\-men\-sion\-al vectors, $x, y \in \reals^t$, their Pearson
correlation coefficient (PCC) is:
\begin{equation}
  \pcc(x,y) = \frac{\sum_{i=1}^t \big(x[i]-\overline{x}\big)\big(y[i]-\overline{y}\big)}{\sqrt{\sum_{i=1}^t\big(x[i]-\overline{x}\big)^2}\sqrt{\sum_{i=1}^t\big(y[i]-\overline{y}\big)^2}}. \label{eq:pcc}
\end{equation}
$\pcc(x,y)$ takes values between $1$ and $-1$, where $1$ means perfect
correlation (positive linear dependence between the entries of $x$ and
$y$), $-1$ means perfect anti-correlation (negative linear dependence
between the entries of $x$ and $y$), and $0$ means a complete lack of
linear correlation.
In rare cases in which in $x$ or $y$
contains only $1$ or only $0$,
the denominator would be equal to $0$,
so we fix the value of $\pcc(x,y)$ as $1$.

\paragraph{(Dis)similarity Measures Between Votes.}
Given two votes, $u$ and $v$, over $m$ candidates,
their Hamming distance is 
\[
    \ham(u,v) = \textstyle \sum_{i \in [m]}|u[i]-v[i]|
\]
and their Jaccard distances is
\[
    \jacc(u,v) =  1-\frac{\sum_{i \in [m]}\min(u[i], v[i])}{\sum_{i \in [m]}\max(u[i],v[i])}.
\]
Hamming distance takes values between $0$ and $m$, whereas Jaccard
takes values between $0$ and $1$. We also sometimes measure similarity
between two votes using Pearson correlation coefficient,
which in the case of binary vectors is known as the phi coefficient~\cite{yul:j:phi-coefficient},
or Matthews coefficient (MCC)~\cite{bal-bru-cha-and-nie:j:classification}.
There are many other (dis)similarity measures between binary vectors
(see, e.g., the discussion of
\citet{dua-str-liu-xu-wu:j:binary-vector-similarity}) but we restrict
our attention to these three as they are common in voting literature.

\section{Election Data}
\label{sec:election_data}
We test our election indices on three kinds of data: 
Special, fixed elections of
particular, well-understood structure, synthetically generated random
elections, and elections obtained from various real-life data (ranging
from actual approval elections to those
derived from other types of preferences).
In the discussion below we consider elections with candidate set
$C = \{c_1, \ldots, c_m\}$ and voter collection
$V = (v_1, \ldots, v_n)$.

\boldpara{Special Elections} Our special elections are:
\begin{description}
\item[$\boldsymbol{p}$-Identity ($\boldsymbol{p}$-ID).] In a $p$-Identity election,
  all voters approve the same set of $\lfloor pm \rfloor$ candidates (e.g., $c_1, \ldots, c_{\lfloor pm \rfloor}$) and
  disapprove the remaining ones.
\item[$\boldsymbol{k}$-Party.] In a $k$-Party election there are $k$
  disjoint, equal-sized groups of candidates ($\pm 1$, depending on
  divisibility) and $k$ disjoint, equal-sized groups of voters
  ($\pm 1$, depending on divisibility) . For each $i \in [k]$, the
  voters in the $i$-th group approve exactly the candidates in the
  $i$-th group. For $x, y \in [0,1]$, by $(x, y)$-$2$-Party we mean a $2$-Party
  election where the first group includes an $x$ fraction of candidates and $y$ fraction of
  voters, and the second group includes the remaining ones.
\end{description}
For the next three elections we require $n=m$:
\begin{description}
\item[Diagonal.] This is an $m$-Party election.
  \item[Triangle.] For each $i \in [n]$, the $i$-th voter approves the first $i$ candidates, i.e., $c_1, \dots, c_i$.
  \item[Cyclic.] The first voter approves the first candidate and
    last $\lfloor m/2 \rfloor-1$ candidates. For each $i \in \{2, \dots n\}$, we obtain the
    $i$-th vote by taking the $(i-1)$-st one and shifting it
    cyclically by one position to the right, i.e.,
    replacing each candidate $c_j$ by $c_{j+1}$,
    with $c_m$ mapped to $c_1$.
\end{description}

\boldpara{Synthetic Elections} Next let us describe our models of
generating random elections. If we only describe a process of sampling
a single vote, it means that all the votes in the election are generated
this way, independently from each other:
\begin{description}
\item[$\boldsymbol{p}$-Impartial Culture ($\boldsymbol{p}$-IC).] For $p \in [0,1]$, we
  sample vote $v$ so that for each $i \in [m]$, we have $v[i] \sim \bernoulli(p)$.

\item[$\boldsymbol{(p_1, \ldots, p_m)}$-Independent Approval Model
  (IAM).]~For $(p_1, \ldots, p_m) \in [0,1]^m$, we sample vote $v$ so
  that for each $i$ $\in [m]$, $v[i] \sim \bernoulli(p_i)$.
  IAMs were discussed, e.g., by~\citet{fal-jan-kac-kur-pie-szu:c:approval-learning},
  \citet{lac-mal:j:shortlisting}, and \citet{xia:t:linear}.

\item[$\boldsymbol{(p,\phi)}$-Resampling.] Given $p, \phi \in [0,1]$
  and a fixed central vote $u$ with $|A(u)| = \lfloor pm\rfloor$, we
  sample vote $v$ as follows: For each $i \in [m]$, with probability
  $1-\phi$ we let $v[i] = u[i]$ and with probability $\phi$ we let
  $v[i] \sim \text{Bernoulli}(p)$. This model is due to
  \citet{szu-fal-jan-lac-sli-sor-tal:c:sampling-approval-elections}.
  Some authors, such as
  \citet{fal-jan-kac-kur-pie-szu:c:approval-learning},
  allow arbitrary central votes, but we use the original variant
  as it gives $pm$ approved candidates in expectation.

  \item[Euclidean.]
  Here, we assume that each voter and candidate is associated with a point in a $d$-dimensional Euclidean space,
  and each voter $v_i$ approves a candidate $c$
  if the distance between their corresponding points
  is smaller than some threshold value $r_i$
  (which can be different for each voter).
  Such a model was studied,
  e.g., by~\citet{god-bat-sko-fal:c:2d}.
  See \Cref{app:dataset} for details.
\end{description}
We also consider a few special distributions, to generate elections
with particular features:
\begin{description}
\item[$\boldsymbol{p}$-ID/IC.] The first half of the voters form a $p$-ID election
  and the second half is generated using the $p$-IC model.
  
\item[Lin-IC.] In this model, $v_1$ is generated using $1$-IC, $v_2$
  using $(1-\nicefrac{1}{n})$-IC, $v_3$ using
  $(1-\nicefrac{2}{n})$-IC, and so on, until $v_n$, generated using
  $\nicefrac{1}{n}$-IC.

\item[$\boldsymbol{\noisy(\calD,\phi)}$.] $\calD$ is a distribution
  over elections. We first generate an election using $\calD$ and then
  we replace each vote $v$ with a new one as follows: We let $p$ be
  the fraction of candidates approved in $v$ and we generate a vote
  using $(p,\phi)$-resampling model with central vote equal to~$v$. In
  other words, $\noisy(\calD,\phi)$ introduces noise on top of
  elections generated using $\calD$ ($\noisy$  stands
  for \emph{noisy}).
    
\end{description}

\boldpara{Real-Life Data}
See \Cref{sec:experiments} and \Cref{app:dataset}.

\section{Indices and Initial Experiments}

By an election index, we mean a function $f$ that given an election
$E$ outputs a value between $0$ and $1$,
quantifying the extent to which $E$ has a feature captured by $f$.
In the following sections we define agreement, diversity, and
polarization indices. For now, we discuss indices on a high level and
describe our approach to evaluating them.

\paragraph{Normalization.}
Given an election $E = (C,V)$, we consider its average vote length
($\avl$; the average number of candidates approved by a voter),
reverse $\avl$ (the average number of candidates disapproved by a
voter), and its saturation ($\satr$; the fraction of candidates that
an average voter approves):
\begin{align*}  
  &\avl(E) =\textstyle \frac{1}{|V|}\sum_{v \in V}|A(v)|,&
  &\revavl(E) = \avl(E^{-1}),&
  &\satr(E) = \avl(E) / |C|.
\end{align*}
While saturation is a natural election index by itself, for the other
indices we are mostly interested in normalizing them in a way that
would be independent of saturation. In other words, if two elections
are similar in nature except that their saturations differ (as would
be the case, e.g., for $0.5$-IC and $0.25$-IC elections), we would
prefer an election index to give them similar values. Naturally, this
is not a precise mathematical statement, but 
rather an intuitive desideratum. 
Yet, in some cases it is easy to argue that a particular
index is not saturation-independent, and we often view it as an
argument against this index.

\newcommand{\resimg}[2]{\begin{tabular}{c}
  \small                          
  \ \ #1\\
\includegraphics[width=1.65cm]{img/agr-resampling/#2\_60\_60.png}
\end{tabular}}

\paragraph{Resampling Experiment.}

\begin{figure}
  \centering
  \vspace{-6mm}
  
  \!\!\!\!\!\begin{tikzpicture}
  \node[inner sep=0] (img) 
    {\includegraphics[width=0.35\linewidth]{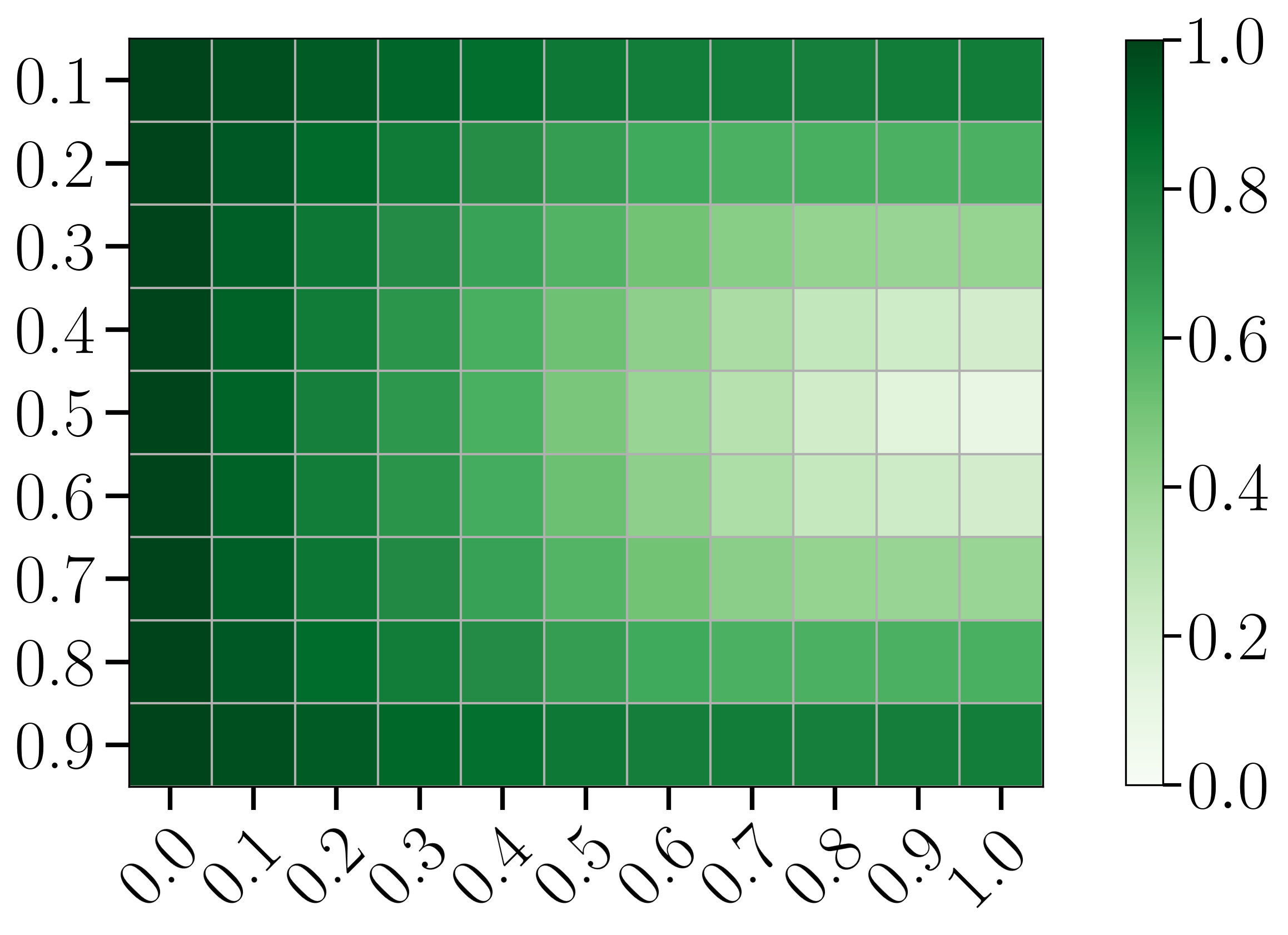}};

  \node[left] at (img.west) {$p$};

  \node[below] at (img.south) {$\phi$};
\end{tikzpicture}
  \caption{\label{fig:resampling}Results of the resampling experiment for Approval Agreement
    index. Rows correspond to fixed $p$ values and columns to fixed $\phi$ values.}  \vspace{-3mm}
\end{figure}

For each of our indices, we perform the following \emph{resampling
  experiment}: {For each $p \in \{0.1, \ldots, 0.9\}$ and each
  $\phi \in \{0, 0.1, \ldots, 1\}$, we generate $10$
  $(p,\phi)$-resampling elections with $60$ candidates and $60$
  voters, and compute for them the average value of the index in
  question.} We arrange the results as a matrix, where each row
corresponds to a fixed value of $p$ and each column corresponds to a
fixed value of $\phi$.  Color brightness represents the average value of the index.
We present such matrices in the first row of
\Cref{tab:index-values}, and we provide an enlarged example (including axis labels and the color-bar scale) in
\Cref{fig:resampling}. The important feature of these plots is that each
column corresponds to a fixed value of $\phi$,
while saturation (i.e., $p$) varies across rows.
Hence, if the columns of the
matrix are not (approximately) uniformly colored, then we claim that a given index is
not saturation-independent (as, arguably, resampling elections with
different values of $p$ but the same value of $\phi$ are of the same
nature). For such an index, we say that it \emph{fails the resampling
test of saturation independence}.

\paragraph{Index Values on Example Elections.}
In \Cref{tab:index-values} we also give values of our indices on a
number of characteristic (distributions over) elections. Each row
gives index values for a given family of elections. All elections have
$60$ candidates and $60$ voters and the results for each family are
averaged over sampling $10$ elections (we report standard deviation in
smaller font). For each family we give its name and an icon presenting
an example election; in this icon, columns are the candidates, rows
are the voters, and darker spots indicate that a given voter approves
a given candidate.
We discuss this table throughout the following sections.

\newcommand{\rsmpl}[1]{\includegraphics[width=1.125cm]{img/agr-resampling/#1\_60\_60.png}}
\newcommand{\resamplingexperiment}{
  \multicolumn{3}{r}{ \raisebox{3mm}{resampling experiment:}} &
  \rsmpl{app\_agr} &
  \rsmpl{cntr\_agr} &
  \rsmpl{pair\_agr} &
  \rsmpl{pcc\_agr} &
  \rsmpl{jacc\_agr} &
  \rsmpl{pccplus\_agr} &
  \rsmpl{cntr-div} &
  \rsmpl{pcc-div} &
  \rsmpl{out-div} &
  \rsmpl{cntr-pol}&
  \rsmpl{pcc-pol} &
  \rsmpl{ham-std} 
  \\
  }

\begin{table*}[t!]
    \small
    \setlength{\tabcolsep}{3pt}
    \scalebox{0.86}{ 
    \small
    \setlength{\tabcolsep}{3pt}
    \begin{tabular}{lclcccccc|ccc|ccc}
    \toprule
&&& \multicolumn{1}{p{0.8cm}}{\rotatebox{53}{$\appragr$}} & \multicolumn{1}{p{0.8cm}}{\rotatebox{53}{$\cntragr$}} & \multicolumn{1}{p{0.8cm}}{\rotatebox{53}{$\pairagr$}} & \multicolumn{1}{p{0.8cm}}{\rotatebox{53}{$\pccagr$}} & \multicolumn{1}{p{0.8cm}}{\rotatebox{53}{$\jaccpairagr$}} & \multicolumn{1}{p{0.8cm}}{\rotatebox{53}{$\pospccagr$}} & \multicolumn{1}{p{0.8cm}}{\rotatebox{53}{cntr-div}} & \multicolumn{1}{p{0.8cm}}{\rotatebox{53}{$\mathrm{pcc}\hbox{-}\mathrm{div}$}} & \multicolumn{1}{p{0.8cm}}{\rotatebox{53}{out-div}} & \multicolumn{1}{p{0.8cm}}{\rotatebox{53}{cntr-pol}} & \multicolumn{1}{p{0.8cm}}{\rotatebox{53}{$\mathrm{pcc}\hbox{-}\mathrm{pol}$}} & \multicolumn{1}{p{0.8cm}}{\rotatebox{53}{$\mathrm{pair}\hbox{-}\mathrm{pol}$}} \\
\midrule
    \resamplingexperiment
    \multicolumn{3}{r}{saturation independence:} &
    \no & \no & \yes & \yes & \no & \yes &
    \no & \yes & \yes &
    -- & -- & -- \\
    \midrule
\raisebox{3mm}{1.} & \raisebox{3mm}{$\nicefrac{1}{3}$-ID} &\includegraphics[width=0.76cm]{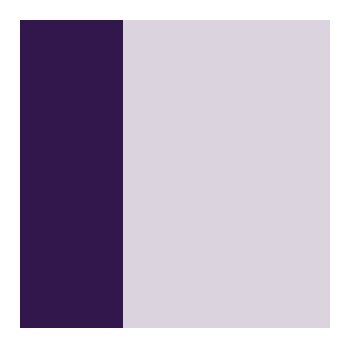}& \raisebox{2mm}{\makecell{\qqagr{1.00}\\ \raisebox{1mm}{\tiny $\pm 0.00$}}}& \raisebox{2mm}{\makecell{\qqagr{1.00}\\ \raisebox{1mm}{\tiny $\pm 0.00$}}}& \raisebox{2mm}{\makecell{\qqagr{1.00}\\ \raisebox{1mm}{\tiny $\pm 0.00$}}}& \raisebox{2mm}{\makecell{\qqagr{1.00}\\ \raisebox{1mm}{\tiny $\pm 0.00$}}}& \raisebox{2mm}{\makecell{\qqagr{1.00}\\ \raisebox{1mm}{\tiny $\pm 0.00$}}}& \raisebox{2mm}{\makecell{\qqagr{1.00}\\ \raisebox{1mm}{\tiny $\pm 0.00$}}}& \raisebox{2mm}{\makecell{\qqdiv{0.00}\\ \raisebox{1mm}{\tiny $\pm 0.00$}}}& \raisebox{2mm}{\makecell{\qqdiv{0.00}\\ \raisebox{1mm}{\tiny $\pm 0.00$}}}& \raisebox{2mm}{\makecell{\qqdiv{0.00}\\ \raisebox{1mm}{\tiny $\pm 0.00$}}}& \raisebox{2mm}{\makecell{\qqpol{0.00}\\ \raisebox{1mm}{\tiny $\pm 0.00$}}}& \raisebox{2mm}{\makecell{\qqpol{0.00}\\ \raisebox{1mm}{\tiny $\pm 0.00$}}}& \raisebox{2mm}{\makecell{\qqpol{0.00}\\ \raisebox{1mm}{\tiny $\pm 0.00$}}}\\[-3pt]
\midrule
\raisebox{3mm}{2.} & \raisebox{3mm}{2-Party} &\includegraphics[width=0.76cm]{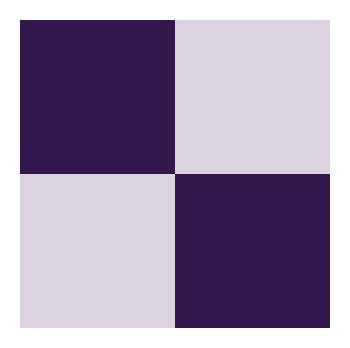}& \raisebox{2mm}{\makecell{\qqagr{0.00}\\ \raisebox{1mm}{\tiny $\pm 0.00$}}}& \raisebox{2mm}{\makecell{\qqagr{0.00}\\ \raisebox{1mm}{\tiny $\pm 0.00$}}}& \raisebox{2mm}{\makecell{\qqagr{0.00}\\ \raisebox{1mm}{\tiny $\pm 0.00$}}}& \raisebox{2mm}{\makecell{\qqagr{0.00}\\ \raisebox{1mm}{\tiny $\pm 0.00$}}}& \raisebox{2mm}{\makecell{\qqagr{0.50}\\ \raisebox{1mm}{\tiny $\pm 0.00$}}}& \raisebox{2mm}{\makecell{\qqagr{0.50}\\ \raisebox{1mm}{\tiny $\pm 0.00$}}}& \raisebox{2mm}{\makecell{\qqdiv{0.20}\\ \raisebox{1mm}{\tiny $\pm 0.00$}}}& \raisebox{2mm}{\makecell{\qqdiv{0.25}\\ \raisebox{1mm}{\tiny $\pm 0.02$}}}& \raisebox{2mm}{\makecell{\qqdiv{0.10}\\ \raisebox{1mm}{\tiny $\pm 0.00$}}}& \raisebox{2mm}{\makecell{\qqpol{1.00}\\ \raisebox{1mm}{\tiny $\pm 0.00$}}}& \raisebox{2mm}{\makecell{\qqpol{1.00}\\ \raisebox{1mm}{\tiny $\pm 0.00$}}}& \raisebox{2mm}{\makecell{\qqpol{1.00}\\ \raisebox{1mm}{\tiny $\pm 0.00$}}}\\[-3pt]
\raisebox{3mm}{3.} & \raisebox{3mm}{N(2-Party, 0.6)} &\includegraphics[width=0.76cm]{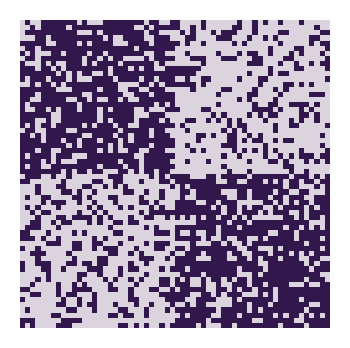}& \raisebox{2mm}{\makecell{\qqagr{0.09}\\ \raisebox{1mm}{\tiny $\pm 0.01$}}}& \raisebox{2mm}{\makecell{\qqagr{0.08}\\ \raisebox{1mm}{\tiny $\pm 0.01$}}}& \raisebox{2mm}{\makecell{\qqagr{0.01}\\ \raisebox{1mm}{\tiny $\pm 0.00$}}}& \raisebox{2mm}{\makecell{\qqagr{0.02}\\ \raisebox{1mm}{\tiny $\pm 0.00$}}}& \raisebox{2mm}{\makecell{\qqagr{0.35}\\ \raisebox{1mm}{\tiny $\pm 0.01$}}}& \raisebox{2mm}{\makecell{\qqagr{0.10}\\ \raisebox{1mm}{\tiny $\pm 0.01$}}}& \raisebox{2mm}{\makecell{\qqdiv{0.66}\\ \raisebox{1mm}{\tiny $\pm 0.02$}}}& \raisebox{2mm}{\makecell{\qqdiv{0.82}\\ \raisebox{1mm}{\tiny $\pm 0.01$}}}& \raisebox{2mm}{\makecell{\qqdiv{0.29}\\ \raisebox{1mm}{\tiny $\pm 0.00$}}}& \raisebox{2mm}{\makecell{\qqpol{0.25}\\ \raisebox{1mm}{\tiny $\pm 0.10$}}}& \raisebox{2mm}{\makecell{\qqpol{0.17}\\ \raisebox{1mm}{\tiny $\pm 0.01$}}}& \raisebox{2mm}{\makecell{\qqpol{0.24}\\ \raisebox{1mm}{\tiny $\pm 0.01$}}}\\[-3pt]
\raisebox{3mm}{4.} & \raisebox{3mm}{3-Party} &\includegraphics[width=0.76cm]{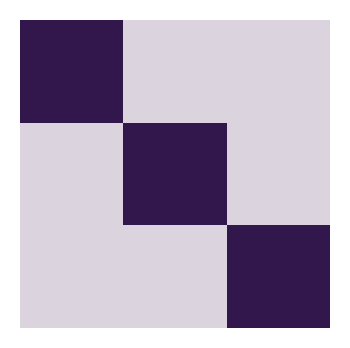}& \raisebox{2mm}{\makecell{\qqagr{0.33}\\ \raisebox{1mm}{\tiny $\pm 0.00$}}}& \raisebox{2mm}{\makecell{\qqagr{0.00}\\ \raisebox{1mm}{\tiny $\pm 0.00$}}}& \raisebox{2mm}{\makecell{\qqagr{0.00}\\ \raisebox{1mm}{\tiny $\pm 0.00$}}}& \raisebox{2mm}{\makecell{\qqagr{0.00}\\ \raisebox{1mm}{\tiny $\pm 0.00$}}}& \raisebox{2mm}{\makecell{\qqagr{0.33}\\ \raisebox{1mm}{\tiny $\pm 0.00$}}}& \raisebox{2mm}{\makecell{\qqagr{0.33}\\ \raisebox{1mm}{\tiny $\pm 0.00$}}}& \raisebox{2mm}{\makecell{\qqdiv{0.33}\\ \raisebox{1mm}{\tiny $\pm 0.00$}}}& \raisebox{2mm}{\makecell{\qqdiv{0.31}\\ \raisebox{1mm}{\tiny $\pm 0.01$}}}& \raisebox{2mm}{\makecell{\qqdiv{0.13}\\ \raisebox{1mm}{\tiny $\pm 0.00$}}}& \raisebox{2mm}{\makecell{\qqpol{0.33}\\ \raisebox{1mm}{\tiny $\pm 0.00$}}}& \raisebox{2mm}{\makecell{\qqpol{0.50}\\ \raisebox{1mm}{\tiny $\pm 0.00$}}}& \raisebox{2mm}{\makecell{\qqpol{0.63}\\ \raisebox{1mm}{\tiny $\pm 0.00$}}}\\[-3pt]
\raisebox{3mm}{5.} & \raisebox{3mm}{4-Party} &\includegraphics[width=0.76cm]{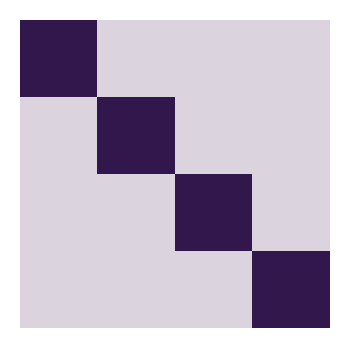}& \raisebox{2mm}{\makecell{\qqagr{0.50}\\ \raisebox{1mm}{\tiny $\pm 0.00$}}}& \raisebox{2mm}{\makecell{\qqagr{0.00}\\ \raisebox{1mm}{\tiny $\pm 0.00$}}}& \raisebox{2mm}{\makecell{\qqagr{0.00}\\ \raisebox{1mm}{\tiny $\pm 0.00$}}}& \raisebox{2mm}{\makecell{\qqagr{0.00}\\ \raisebox{1mm}{\tiny $\pm 0.00$}}}& \raisebox{2mm}{\makecell{\qqagr{0.25}\\ \raisebox{1mm}{\tiny $\pm 0.00$}}}& \raisebox{2mm}{\makecell{\qqagr{0.25}\\ \raisebox{1mm}{\tiny $\pm 0.00$}}}& \raisebox{2mm}{\makecell{\qqdiv{0.47}\\ \raisebox{1mm}{\tiny $\pm 0.04$}}}& \raisebox{2mm}{\makecell{\qqdiv{0.40}\\ \raisebox{1mm}{\tiny $\pm 0.00$}}}& \raisebox{2mm}{\makecell{\qqdiv{0.15}\\ \raisebox{1mm}{\tiny $\pm 0.00$}}}& \raisebox{2mm}{\makecell{\qqpol{0.25}\\ \raisebox{1mm}{\tiny $\pm 0.00$}}}& \raisebox{2mm}{\makecell{\qqpol{0.33}\\ \raisebox{1mm}{\tiny $\pm 0.00$}}}& \raisebox{2mm}{\makecell{\qqpol{0.43}\\ \raisebox{1mm}{\tiny $\pm 0.00$}}}\\[-3pt]
\raisebox{3mm}{6.} & \raisebox{3mm}{$(\nicefrac{1}{3}, \nicefrac{1}{3})$-2-Party} &\includegraphics[width=0.76cm]{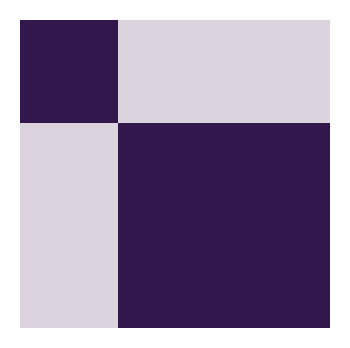}& \raisebox{2mm}{\makecell{\qqagr{0.33}\\ \raisebox{1mm}{\tiny $\pm 0.00$}}}& \raisebox{2mm}{\makecell{\qqagr{0.24}\\ \raisebox{1mm}{\tiny $\pm 0.00$}}}& \raisebox{2mm}{\makecell{\qqagr{0.10}\\ \raisebox{1mm}{\tiny $\pm 0.00$}}}& \raisebox{2mm}{\makecell{\qqagr{0.11}\\ \raisebox{1mm}{\tiny $\pm 0.00$}}}& \raisebox{2mm}{\makecell{\qqagr{0.56}\\ \raisebox{1mm}{\tiny $\pm 0.00$}}}& \raisebox{2mm}{\makecell{\qqagr{0.56}\\ \raisebox{1mm}{\tiny $\pm 0.00$}}}& \raisebox{2mm}{\makecell{\qqdiv{0.15}\\ \raisebox{1mm}{\tiny $\pm 0.00$}}}& \raisebox{2mm}{\makecell{\qqdiv{0.18}\\ \raisebox{1mm}{\tiny $\pm 0.00$}}}& \raisebox{2mm}{\makecell{\qqdiv{0.09}\\ \raisebox{1mm}{\tiny $\pm 0.00$}}}& \raisebox{2mm}{\makecell{\qqpol{0.76}\\ \raisebox{1mm}{\tiny $\pm 0.00$}}}& \raisebox{2mm}{\makecell{\qqpol{0.89}\\ \raisebox{1mm}{\tiny $\pm 0.00$}}}& \raisebox{2mm}{\makecell{\qqpol{0.99}\\ \raisebox{1mm}{\tiny $\pm 0.00$}}}\\[-3pt]
\raisebox{3mm}{7.} & \raisebox{3mm}{Cyclic} &\includegraphics[width=0.76cm]{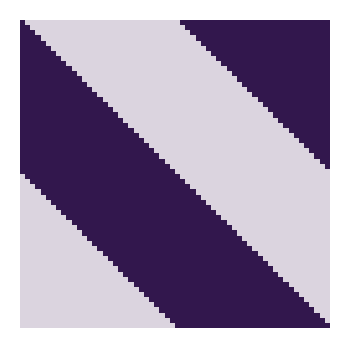}& \raisebox{2mm}{\makecell{\qqagr{0.00}\\ \raisebox{1mm}{\tiny $\pm 0.00$}}}& \raisebox{2mm}{\makecell{\qqagr{0.00}\\ \raisebox{1mm}{\tiny $\pm 0.00$}}}& \raisebox{2mm}{\makecell{\qqagr{0.00}\\ \raisebox{1mm}{\tiny $\pm 0.00$}}}& \raisebox{2mm}{\makecell{\qqagr{0.00}\\ \raisebox{1mm}{\tiny $\pm 0.00$}}}& \raisebox{2mm}{\makecell{\qqagr{0.39}\\ \raisebox{1mm}{\tiny $\pm 0.00$}}}& \raisebox{2mm}{\makecell{\qqagr{0.25}\\ \raisebox{1mm}{\tiny $\pm 0.00$}}}& \raisebox{2mm}{\makecell{\qqdiv{0.46}\\ \raisebox{1mm}{\tiny $\pm 0.00$}}}& \raisebox{2mm}{\makecell{\qqdiv{0.57}\\ \raisebox{1mm}{\tiny $\pm 0.01$}}}& \raisebox{2mm}{\makecell{\qqdiv{0.22}\\ \raisebox{1mm}{\tiny $\pm 0.00$}}}& \raisebox{2mm}{\makecell{\qqpol{0.50}\\ \raisebox{1mm}{\tiny $\pm 0.00$}}}& \raisebox{2mm}{\makecell{\qqpol{0.33}\\ \raisebox{1mm}{\tiny $\pm 0.00$}}}& \raisebox{2mm}{\makecell{\qqpol{0.58}\\ \raisebox{1mm}{\tiny $\pm 0.00$}}}\\[-3pt]
\midrule
\raisebox{3mm}{8.} & \raisebox{3mm}{Diagonal} &\includegraphics[width=0.76cm]{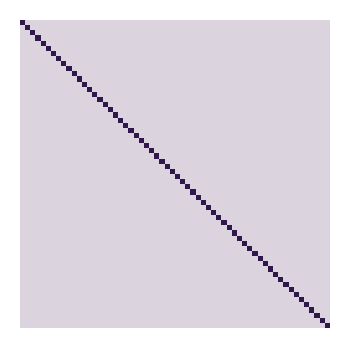}& \raisebox{2mm}{\makecell{\qqagr{0.97}\\ \raisebox{1mm}{\tiny $\pm 0.00$}}}& \raisebox{2mm}{\makecell{\qqagr{0.00}\\ \raisebox{1mm}{\tiny $\pm 0.00$}}}& \raisebox{2mm}{\makecell{\qqagr{0.00}\\ \raisebox{1mm}{\tiny $\pm 0.00$}}}& \raisebox{2mm}{\makecell{\qqagr{0.00}\\ \raisebox{1mm}{\tiny $\pm 0.00$}}}& \raisebox{2mm}{\makecell{\qqagr{0.02}\\ \raisebox{1mm}{\tiny $\pm 0.00$}}}& \raisebox{2mm}{\makecell{\qqagr{0.02}\\ \raisebox{1mm}{\tiny $\pm 0.00$}}}& \raisebox{2mm}{\makecell{\qqdiv{0.97}\\ \raisebox{1mm}{\tiny $\pm 0.00$}}}& \raisebox{2mm}{\makecell{\qqdiv{0.97}\\ \raisebox{1mm}{\tiny $\pm 0.00$}}}& \raisebox{2mm}{\makecell{\qqdiv{0.63}\\ \raisebox{1mm}{\tiny $\pm 0.01$}}}& \raisebox{2mm}{\makecell{\qqpol{0.02}\\ \raisebox{1mm}{\tiny $\pm 0.00$}}}& \raisebox{2mm}{\makecell{\qqpol{0.02}\\ \raisebox{1mm}{\tiny $\pm 0.00$}}}& \raisebox{2mm}{\makecell{\qqpol{0.01}\\ \raisebox{1mm}{\tiny $\pm 0.00$}}}\\[-3pt]
\raisebox{3mm}{9.} & \raisebox{3mm}{Triangle} &\includegraphics[width=0.76cm]{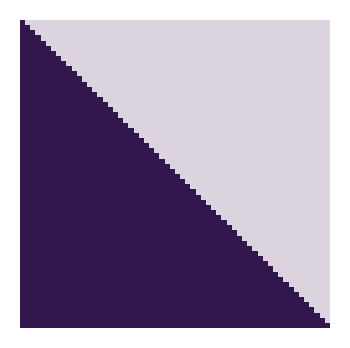}& \raisebox{2mm}{\makecell{\qqagr{0.50}\\ \raisebox{1mm}{\tiny $\pm 0.00$}}}& \raisebox{2mm}{\makecell{\qqagr{0.49}\\ \raisebox{1mm}{\tiny $\pm 0.00$}}}& \raisebox{2mm}{\makecell{\qqagr{0.33}\\ \raisebox{1mm}{\tiny $\pm 0.00$}}}& \raisebox{2mm}{\makecell{\qqagr{0.50}\\ \raisebox{1mm}{\tiny $\pm 0.00$}}}& \raisebox{2mm}{\makecell{\qqagr{0.51}\\ \raisebox{1mm}{\tiny $\pm 0.00$}}}& \raisebox{2mm}{\makecell{\qqagr{0.50}\\ \raisebox{1mm}{\tiny $\pm 0.00$}}}& \raisebox{2mm}{\makecell{\qqdiv{0.41}\\ \raisebox{1mm}{\tiny $\pm 0.00$}}}& \raisebox{2mm}{\makecell{\qqdiv{0.32}\\ \raisebox{1mm}{\tiny $\pm 0.00$}}}& \raisebox{2mm}{\makecell{\qqdiv{0.18}\\ \raisebox{1mm}{\tiny $\pm 0.00$}}}& \raisebox{2mm}{\makecell{\qqpol{0.01}\\ \raisebox{1mm}{\tiny $\pm 0.00$}}}& \raisebox{2mm}{\makecell{\qqpol{0.14}\\ \raisebox{1mm}{\tiny $\pm 0.00$}}}& \raisebox{2mm}{\makecell{\qqpol{0.47}\\ \raisebox{1mm}{\tiny $\pm 0.00$}}}\\[-3pt]
\raisebox{3mm}{10.} & \raisebox{3mm}{N(Tri,0.6)} &\includegraphics[width=0.76cm]{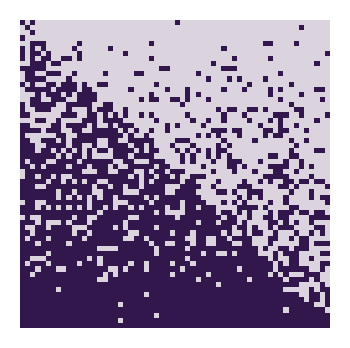}& \raisebox{2mm}{\makecell{\qqagr{0.21}\\ \raisebox{1mm}{\tiny $\pm 0.01$}}}& \raisebox{2mm}{\makecell{\qqagr{0.20}\\ \raisebox{1mm}{\tiny $\pm 0.02$}}}& \raisebox{2mm}{\makecell{\qqagr{0.06}\\ \raisebox{1mm}{\tiny $\pm 0.00$}}}& \raisebox{2mm}{\makecell{\qqagr{0.15}\\ \raisebox{1mm}{\tiny $\pm 0.02$}}}& \raisebox{2mm}{\makecell{\qqagr{0.33}\\ \raisebox{1mm}{\tiny $\pm 0.01$}}}& \raisebox{2mm}{\makecell{\qqagr{0.17}\\ \raisebox{1mm}{\tiny $\pm 0.02$}}}& \raisebox{2mm}{\makecell{\qqdiv{0.89}\\ \raisebox{1mm}{\tiny $\pm 0.01$}}}& \raisebox{2mm}{\makecell{\qqdiv{0.77}\\ \raisebox{1mm}{\tiny $\pm 0.03$}}}& \raisebox{2mm}{\makecell{\qqdiv{0.26}\\ \raisebox{1mm}{\tiny $\pm 0.00$}}}& \raisebox{2mm}{\makecell{\qqpol{0.00}\\ \raisebox{1mm}{\tiny $\pm 0.02$}}}& \raisebox{2mm}{\makecell{\qqpol{0.05}\\ \raisebox{1mm}{\tiny $\pm 0.00$}}}& \raisebox{2mm}{\makecell{\qqpol{0.39}\\ \raisebox{1mm}{\tiny $\pm 0.01$}}}\\[-3pt]
\raisebox{3mm}{11.} & \raisebox{3mm}{$\nicefrac{1}{2}$-ID/IC} &\includegraphics[width=0.76cm]{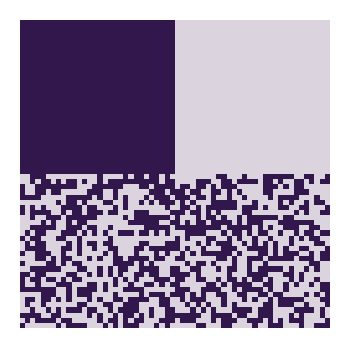}& \raisebox{2mm}{\makecell{\qqagr{0.50}\\ \raisebox{1mm}{\tiny $\pm 0.01$}}}& \raisebox{2mm}{\makecell{\qqagr{0.49}\\ \raisebox{1mm}{\tiny $\pm 0.01$}}}& \raisebox{2mm}{\makecell{\qqagr{0.26}\\ \raisebox{1mm}{\tiny $\pm 0.01$}}}& \raisebox{2mm}{\makecell{\qqagr{0.25}\\ \raisebox{1mm}{\tiny $\pm 0.01$}}}& \raisebox{2mm}{\makecell{\qqagr{0.51}\\ \raisebox{1mm}{\tiny $\pm 0.00$}}}& \raisebox{2mm}{\makecell{\qqagr{0.30}\\ \raisebox{1mm}{\tiny $\pm 0.00$}}}& \raisebox{2mm}{\makecell{\qqdiv{0.41}\\ \raisebox{1mm}{\tiny $\pm 0.01$}}}& \raisebox{2mm}{\makecell{\qqdiv{0.50}\\ \raisebox{1mm}{\tiny $\pm 0.00$}}}& \raisebox{2mm}{\makecell{\qqdiv{0.19}\\ \raisebox{1mm}{\tiny $\pm 0.00$}}}& \raisebox{2mm}{\makecell{\qqpol{0.07}\\ \raisebox{1mm}{\tiny $\pm 0.01$}}}& \raisebox{2mm}{\makecell{\qqpol{0.26}\\ \raisebox{1mm}{\tiny $\pm 0.02$}}}& \raisebox{2mm}{\makecell{\qqpol{0.45}\\ \raisebox{1mm}{\tiny $\pm 0.01$}}}\\[-3pt]
\raisebox{3mm}{12.} & \raisebox{3mm}{$\nicefrac{1}{2}$-IC} &\includegraphics[width=0.76cm]{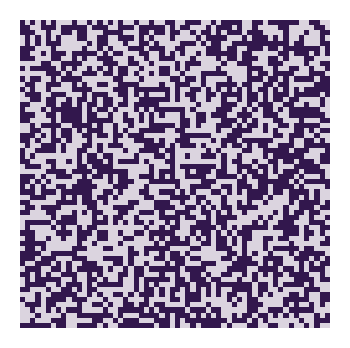}& \raisebox{2mm}{\makecell{\qqagr{0.10}\\ \raisebox{1mm}{\tiny $\pm 0.01$}}}& \raisebox{2mm}{\makecell{\qqagr{0.10}\\ \raisebox{1mm}{\tiny $\pm 0.01$}}}& \raisebox{2mm}{\makecell{\qqagr{0.02}\\ \raisebox{1mm}{\tiny $\pm 0.00$}}}& \raisebox{2mm}{\makecell{\qqagr{0.02}\\ \raisebox{1mm}{\tiny $\pm 0.00$}}}& \raisebox{2mm}{\makecell{\qqagr{0.34}\\ \raisebox{1mm}{\tiny $\pm 0.01$}}}& \raisebox{2mm}{\makecell{\qqagr{0.07}\\ \raisebox{1mm}{\tiny $\pm 0.00$}}}& \raisebox{2mm}{\makecell{\qqdiv{0.81}\\ \raisebox{1mm}{\tiny $\pm 0.01$}}}& \raisebox{2mm}{\makecell{\qqdiv{0.91}\\ \raisebox{1mm}{\tiny $\pm 0.00$}}}& \raisebox{2mm}{\makecell{\qqdiv{0.29}\\ \raisebox{1mm}{\tiny $\pm 0.00$}}}& \raisebox{2mm}{\makecell{\qqpol{0.07}\\ \raisebox{1mm}{\tiny $\pm 0.01$}}}& \raisebox{2mm}{\makecell{\qqpol{0.05}\\ \raisebox{1mm}{\tiny $\pm 0.00$}}}& \raisebox{2mm}{\makecell{\qqpol{0.18}\\ \raisebox{1mm}{\tiny $\pm 0.00$}}}\\[-3pt]
\raisebox{3mm}{13.} & \raisebox{3mm}{$\nicefrac{1}{4}$-IC} &\includegraphics[width=0.76cm]{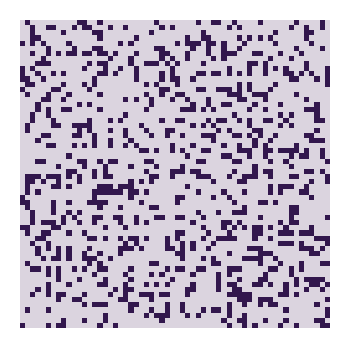}& \raisebox{2mm}{\makecell{\qqagr{0.49}\\ \raisebox{1mm}{\tiny $\pm 0.01$}}}& \raisebox{2mm}{\makecell{\qqagr{0.00}\\ \raisebox{1mm}{\tiny $\pm 0.00$}}}& \raisebox{2mm}{\makecell{\qqagr{0.02}\\ \raisebox{1mm}{\tiny $\pm 0.00$}}}& \raisebox{2mm}{\makecell{\qqagr{0.02}\\ \raisebox{1mm}{\tiny $\pm 0.00$}}}& \raisebox{2mm}{\makecell{\qqagr{0.16}\\ \raisebox{1mm}{\tiny $\pm 0.00$}}}& \raisebox{2mm}{\makecell{\qqagr{0.07}\\ \raisebox{1mm}{\tiny $\pm 0.00$}}}& \raisebox{2mm}{\makecell{\qqdiv{0.97}\\ \raisebox{1mm}{\tiny $\pm 0.00$}}}& \raisebox{2mm}{\makecell{\qqdiv{0.90}\\ \raisebox{1mm}{\tiny $\pm 0.00$}}}& \raisebox{2mm}{\makecell{\qqdiv{0.31}\\ \raisebox{1mm}{\tiny $\pm 0.00$}}}& \raisebox{2mm}{\makecell{\qqpol{0.01}\\ \raisebox{1mm}{\tiny $\pm 0.01$}}}& \raisebox{2mm}{\makecell{\qqpol{0.05}\\ \raisebox{1mm}{\tiny $\pm 0.00$}}}& \raisebox{2mm}{\makecell{\qqpol{0.16}\\ \raisebox{1mm}{\tiny $\pm 0.00$}}}\\[-3pt]
\raisebox{3mm}{14.} & \raisebox{3mm}{Lin-IC} &\includegraphics[width=0.76cm]{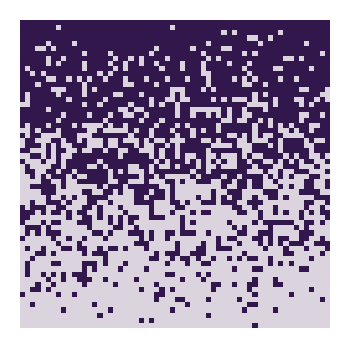}& \raisebox{2mm}{\makecell{\qqagr{0.08}\\ \raisebox{1mm}{\tiny $\pm 0.01$}}}& \raisebox{2mm}{\makecell{\qqagr{0.06}\\ \raisebox{1mm}{\tiny $\pm 0.01$}}}& \raisebox{2mm}{\makecell{\qqagr{0.01}\\ \raisebox{1mm}{\tiny $\pm 0.00$}}}& \raisebox{2mm}{\makecell{\qqagr{0.09}\\ \raisebox{1mm}{\tiny $\pm 0.02$}}}& \raisebox{2mm}{\makecell{\qqagr{0.31}\\ \raisebox{1mm}{\tiny $\pm 0.01$}}}& \raisebox{2mm}{\makecell{\qqagr{0.15}\\ \raisebox{1mm}{\tiny $\pm 0.03$}}}& \raisebox{2mm}{\makecell{\qqdiv{0.96}\\ \raisebox{1mm}{\tiny $\pm 0.01$}}}& \raisebox{2mm}{\makecell{\qqdiv{0.83}\\ \raisebox{1mm}{\tiny $\pm 0.03$}}}& \raisebox{2mm}{\makecell{\qqdiv{0.27}\\ \raisebox{1mm}{\tiny $\pm 0.00$}}}& \raisebox{2mm}{\makecell{\qqpol{0.00}\\ \raisebox{1mm}{\tiny $\pm 0.01$}}}& \raisebox{2mm}{\makecell{\qqpol{0.04}\\ \raisebox{1mm}{\tiny $\pm 0.00$}}}& \raisebox{2mm}{\makecell{\qqpol{0.37}\\ \raisebox{1mm}{\tiny $\pm 0.01$}}}\\[-3pt]

    \bottomrule
    \end{tabular}
    }
    \caption{\label{tab:index-values}Values of our indices for elections from various distributions. The elections have $60$ candidates and $60$ voters. Index
      values are averaged over 10 samples and the standard deviation is reported in smaller font. Similar table with more elections can be found in \Cref{app:dataset}.}
\end{table*}

\section{Agreement Indices}
\label{sec:agreement}
An agreement index should identify whether the voters have similar
views, i.e., tend to approve and disapprove the same candidates.  We
consider two types of such indices: Global ones look at the extent to
which all voters agree with each other, whereas local ones also detect
large groups of voters that are internally in agreement, but that may
disagree with each other.
We propose six agreement indices that we describe below.
All our indices give value $1$ exactly for $p$-ID elections---which,
indeed, represent perfect agreements---and for most of them we also
characterize when they give value~$0$.

\subsection{Approval Agreement}\label{sec:aa}
Our first index, Approval Agreement, adapts an idea from the world of
ordinal elections (where the voters rank the candidates).
There, for every pair of candidates $a$ and
$b$ we compute the absolute difference between the fractions of voters
that rank them one way or the other, and output the average of these
values~\citep{alc-vor:j:cohesiveness,has-end:c:diversity-indices,can-ozk-sto:polarization,fal-kac-sor-szu-was:c:div-agr-pol-map}.
In the approval setting, we compare the numbers of voters approving
and disapproving each given candidate.
\begin{definition}\label{def:av-index}
  Given an election $E = (C,V)$, its Approval
  Agreement is 
  \[
    \appragr(E) = 
    \frac{1}{|C|}\sum_{c \in C}\left|\frac{|A(c)|}{|V|} - \left(1-\frac{|A(c)|}{|V|}\right)\right| =
    \frac{1}{|C|}\sum_{c \in C}\left|1-2\frac{|A(c)|}{|V|}\right|.
    \]
\end{definition}
Unfortunately, this index has several drawbacks. Foremost, as we see
in the first row of \Cref{tab:index-values}, it fails the resampling
test of saturation independence (and, indeed, we see that it gives
different values for $\nicefrac{1}{2}$-IC and $\nicefrac{1}{4}$-IC).
This is also visible in the following proposition, where we show that
it assumes value $0$ if and only if each candidate is approved by
exactly half of the voters and, hence, it is not possible to assume
this value for saturation different from $0.5$.

\begin{proposition}\label{pro:av-agr-characterization}
  For an election $E$, $\appragr(E) = 1$ exactly if $E$ is an identity
  election, and $\appragr(E) = 0$ exactly if each candidate is
  approved by exactly half of the voters.
\end{proposition}

The index also behaves poorly on Party elections. As we see in
\Cref{tab:index-values}, for $2$-Party its value is $0$, but it
increases as we consider $3$-Party and $4$-Party, to become close to
$1$ for Diagonal. Indeed, the index treats the fact that many voters
disapprove the same candidates as a sign of agreement among them
(while in some cases---such as elections within parliaments---this
indeed might be a reasonable view, in others---such as in
participatory budgeting---it is ungrounded).

\subsection{Central Agreement(s)}
Our second index, Central Agreement, is designed to resemble the
Kemeny voting rule for ordinal elections~\citep{kem:j:no-numbers}.  In
that setting, a Kemeny ranking is a ranking that minimizes the sum of
swap distances to all the votes in a given election. The larger is
this value, the bigger is the disagreement among the voters. We
implement the same idea for approval elections using the Hamming
distance (while computing a Kemeny ranking is well-known to be
intractable~\citep{bar-tov-tri:j:who-won,hem-spa-vog:j:kemeny}, in our
case the analogous task is very simple).

Let $E = (C,V)$ be an election.
A vote $u$ over candidate set $C$ is central for $E$ if for each
candidate $c \in C$, $u$ agrees regarding $c$ with at least half of
the voters.
So, if exactly half of the voters approve some candidate and half
disapprove it, then there is more than one central vote for the given
election. We write $\cen(E)$ to denote the set of all central votes
for $E$.
\begin{definition}
  \label{def:cntr-agr}
  Let $E = (C,V)$ be an election and let $u$ be an
  arbitrary member of $\cen(E)$. We define
  \[ \cntragr(E)= 1 - \frac{\sum_{v_i \in V} \ham(v_i, u)}{|V| \cdot
    \min(\avl(E), \revavl(E))}.\]
\end{definition}
The denominator in the definition of $\cntragr(E)$ is chosen to be
(nearly) equal to the largest possible value of the numerator for
elections of the same size and saturation as $E$.
\begin{proposition}\label{pro:central-characterization}
  For an election $E$, $\cntragr(E) = 1$ exactly if $E$ is an identity
  election, and $\cntragr(E) = 0$ exactly if either $\cen(E)$ contains a vote
  approving all candidates or $\cen(E)$ contains a vote disapproving all candidates.
\end{proposition}

Central Agreement has two major drawbacks. First, as indicated by
\Cref{pro:central-characterization}, it is quite non-discriminative
and often assumes value $0$ (e.g., on the majority of Pabulib elections).
Second, it fails the resampling test of saturation independence (see
\Cref{tab:index-values}).
Interestingly, 
Central Agreement and Approval Agreement differ mostly by their
normalization (replacing the denominator with $m$ in
\Cref{pro:cntr-value} below, we would get Approval Agreement).
\begin{proposition}\label{pro:cntr-value}
  Let $E = (C,V)$, then
  $$\cntragr(E) = 1 - \sum_{c \in C}
  \frac{1-\left|1-\frac{2|A(c)|}{|V|}\right|}{2 \cdot \min(\avl(E), \revavl(E))}.$$
\end{proposition}

We do not consider central agreement based on the Jaccard distance as
it would be intractable~\citep{chi-kum-pan-vas:c:jaccard-median} and
we find sufficiently many other agreement indices.

\subsection{Pairwise Agreements}
The final four indices are based on computing expected (dis)similarity
between two randomly selected votes.
\begin{definition}\label{def:pair-agree-indices}
  For an election $E = (C,V)$, we define Hamming, Jaccard, PCC, and
  PCC$^+$ Pairwise Agreement Indices:
  \begin{align*}
    \pairagr(E)     &= \textstyle 1-\frac{\sum_{u \in V}\sum_{v \in V} \ham(u,v)}{2|V|^2 |C| \cdot \satr(E)(1-\satr(E))},\\
    \jaccpairagr(E) &= \textstyle\frac{1}{|V|^2}{\sum_{u \in V}\sum_{v \in V}(1-\jacc(u,v))},\\
    \pccagr(E)      &= \textstyle\frac{1}{|V|^2}\sum_{u \in V}\sum_{v \in V}\pcc(u,v),\\
    \pospccagr(E)   &= \textstyle\frac{1}{|V|^2}\sum_{u \in V}\sum_{v \in V}\max(0,\pcc(u,v)).    
  \end{align*}
\end{definition}
The normalization for Hamming Pairwise Agreement
follows that for Central Agreement:
the denominator is equal to the maximum possible value of the
numerator among elections with the same size and saturation (assuming
that it is possible to ensure that every candidate has the same
approval score).
The normalizations for Jaccard and PCC$^+$ use the fact that we
average over values that
are between $0$ and $1$. For PCC Agreement, we need an additional argument.
\begin{proposition}\label{pro:pcc-is-index}
  For each election $E$, $\pccagr(E) \in [0,1]$.
\end{proposition}
At first sight, computing each of these indices requires $O(|C||V|^2)$
time: Computing similarity between a pair of votes requires $O(|C|)$
time and there are $O(|V|^2)$ pairs to consider. Yet, by changing the
order of summation, we can compute Hamming Pairwise Agreement 
in $O(|C||V|)$ time.
\begin{theorem}\label{thm:pairwise-value}
  For an election $E = (C,V)$, we have that
  $$\pairagr(E) = 1-\frac{\sum_{c \in C}
    |A(c)|\cdot(|V|-|A(c)|) }{|V|^2 |C| \cdot
    \satr(E)(1-\satr(E))}.$$
\end{theorem}

Next we characterize when Hamming and Jaccard Pairwise indices assume
values $1$ and $0$. In case of PCC-based indices we
only provide some examples and observations.
\begin{theorem}\label{thm:pairwise-characterization}
  Each of Pairwise, Jaccard, PCC, and PCC$^+$ Agreement Indices
  assumes value $1$ exactly for identity elections. Hamming Pairwise
  Agreement has value $0$ exactly if all candidates have equal
  approval score, and Jaccard Pairwise Agreement has value $0$ exactly
  if all the voters approve disjoint sets of candidates.
\end{theorem}
\begin{proposition}\label{pro:pcc}
  For each $k$-Party election $E = (C,V)$ such that both $|C|$ and
  $|V|$ are divisible by $k$, we have $\pccagr(E) = 0$ and
  $\pospccagr(E) = \nicefrac{1}{k}$.
\end{proposition}
\begin{remark}
  Let $p \in [0,1]$ be a number and let $u$ and $v$ be two votes such
  that a $p^2$ fraction of candidates is approved in both votes,
  a $(1-p)^2$ fraction is disapproved in both, a $p(1-p)$
  fraction is only approved in $u$, and a $p(1-p)$ fraction is only
  approved in $v$. Then $\pcc(u,v) = 0$.  As votes in $p$-IC elections
  are close to satisfying these conditions, we expect PCC and PCC$^+$
  Agreements to be close to $0$ for IC elections.
\end{remark}

We note that Hamming Pairwise Agreement and PCC Agreement tend to give
very similar results for the examples in \Cref{tab:index-values}.  In
fact, if all voters approve the same number of candidates, then the
two indices are equal (but this is not the case when vote lengths
differ; see the Triangle election).
\begin{proposition}\label{prop:pairagr-pccagr-sometimes-equal}
  Let $E=(C,V)$ be an election such that $0 < |A(u)| = |A(v)| < |C|$ for all $u, v \in V$.
  Then we have that $\pairagr(E) = \pccagr(E)$.
\end{proposition}
In the resampling experiment both Hamming and PCC Agreements give
results independent of saturation. Overall, we view both of them
as good global agreement indices.

Jaccard and PCC$^+$ Pairwise Agreement indices capture local
agreement, by detecting 
large agreeing groups, even if they disagree with other such groups
(indeed, in PCC$^+$ we only sum positive PCC values to account for
like-minded voters, disregarding negative correlations that could
cancel them out).  For example, they give high---and nearly
identical---values for party elections (the fewer parties, the higher
agreement, as there are larger agreeing groups) and detect agreeing
groups in $\nicefrac{1}{2}$-$\ID$/$\IC$,
and Lin-IC.
Yet, we prefer PCC$^+$ Agreement as it passes the resampling test of
saturation independence and assigns close-to-zero values to $\IC$
elections.  The fact that Jaccard Pairwise Agreement fails the
resampling test of saturation is expected, as it treats approvals and
disapprovals assymetrically.  If such asymmetry is natural in a given
setting it still might be a valuable index.

\section{Diversity and Polarization Indices}\label{sec:div-pol}
Intuitively, an election is diverse if (a) its votes tend to be
different from each other, and (b) they cover the space of all
possible votes (modulo saturation). For example, we view Diagonal as a
maximally diverse election, because for its saturation (one approval
per voter) it contains exactly one copy of each possible vote. Yet, we
view Cyclic, which is similar in spirit to Diagonal, as much less
diverse than $\nicefrac{1}{2}$-IC; these two types of election have
the same saturation, but the votes in the latter are notably more
different from each other and, intuitively, cover the preference space
more evenly.

On the other hand, we say that an election is polarized if it
consists of two groups of voters that agree internally but strongly
disagree with each other. Consequently, $2$-Party is an example of
a maximally polarized election. We acknowledge that elections with
three, four, or more internally consistent but disagreeing groups
could also be seen as maximally polarized---and seeking indices
capturing this intuition would be valuable---but we leave this for
future work.

For a broader view of diversity and polarization, see the special
issue edited by \citet{lev-mil-perr:dynamics-of-polit-polar}.  In
voting literature, see, e.g., the works of
\citet{neh-pup:j:diversity}, \citet{has-end:c:diversity-indices},
\citet{can-ozk-sto:polarization},
\citet{kar-mar-rii-zho:t:condorcet-diversity},
\citet{amm-pup:j:domain-diversity}, and
\citeauthor{fal-kac-sor-szu-was:c:div-agr-pol-map}~[\citeyear{fal-kac-sor-szu-was:c:div-agr-pol-map,fal-mer-nun-szu-was:c:diff-sizes,fal-was-sor-szu:c:k-kemeny,fal-was-sor-szu:c:outer-diversity}].

\subsection{Clustering-Based Diversity and Polarization}
Our first approach to designing diversity indices follows the view of
\citet{fal-kac-sor-szu-was:c:div-agr-pol-map,fal-was-sor-szu:c:k-kemeny},
that an election is diverse if even after clustering its
voters into like-minded groups, these groups disagree internally.
As opposed to them, we parameterize our definition with
an agreement index.

For a given integer $k$ and an election $E = (C,V)$, let
$\mathcal{P}_k(V)$ be a set of all complete partitions of voters into
$k$ disjoint subsets $V_1,\dots,V_k$, and let $\agr$ be some agreement
index.  We aim to find a partition that maximizes average agreement
within the subsets, weighted by the size of each cluster, i.e.,
$\sum_{i=1}^k \frac{|V_i|}{|V|}\agr(C,V_i)$.  In principle, there are
many ways in which the values for different $k$ can be aggregated to
give a single index.  We follow a simple approach of
\citet{fal-mer-nun-szu-was:c:diff-sizes} and take the average of
values for $k$ from $1$ to $5$.  This results in the following
definition.

\begin{definition}
\label{def:inner:diversity}
  For an election $E = (C,V)$, and agreement index $\agr$,
  we define the $\agr$-\emph{Diversity} index as:
    \[
    \textstyle 
    \agr\hbox{-}\mathrm{div}(E) = 1 -\frac{1}{5}\sum_{k = 1}^5 \displaystyle \max_{\{V_1,\dots,V_k\} \in \mathcal{P}_k(V)} \textstyle \sum_{i=1}^k \frac{|V_i|}{|V|}\agr(C,V_i).
    \]
\end{definition}
Taking Central Agreement as the underlying index, we obtain Central
Diversity, denoted $\cntrdiv$, and taking PCC Pairwise Agreement we
obtain PCC Diversity, denoted $\pccdiv$. We do not consider Hamming
Pairwise Agreement for the sake of focus and due to its
similarity to $\pccagr$.
In both cases, obtaining the optimal partition of votes is not
computationally feasible, thus we rely on clustering heuristics.  We
approximate cntr-div with an algorithm based on the $k$-means
clustering: We employ the $k$-medoids version using Hamming distance,
where cluster centers are restricted to observed votes and are
iteratively updated to minimize the sum of intra-cluster
distances. For pcc-div we use spectral
clustering~\citep{ng-jor-wei:c:spectral-clustering}; our input for
spectral clustering is the number $k$ of clusters and the affinity
matrix, where for each pair of votes $u, v$ we provide value
$1+\frac{1}{2}\pcc(u,v)$ (value~$0$ means total dissimilarity, and $1$
means the votes are equal).

To quantify polarization, we measure the increase of agreement arising
from clustering the election into two groups instead of viewing it as
a single group (again, the general idea is due to
\citet{fal-kac-sor-szu-was:c:div-agr-pol-map,fal-was-sor-szu:c:k-kemeny}).
\begin{definition}
\label{def:agr:based:polarization}
  For an election $E = (C,V)$, and agreement index $\agr$,
  we define the $\agr$-\emph{Polarization} index as:
    \[
    \textstyle 
    \agr\hbox{-}\mathrm{pol}(E) = \displaystyle \max_{\{V_1,V_2\} \in \mathcal{P}_2(V)} \textstyle \sum_{i=1}^2\frac{|V_i|}{|V|}\agr(C,V_i) - \agr(E).
    \]
\end{definition}
We consider Central Polarization ($\cntrpol$) and PCC Polarization
($\pccpol$), based on the Central and PCC Pairwise Agreement indices,
and we use the same clustering heuristics as in the case of diversity
indices.

Additionally, we consider the Pairwise Polarization index, which
outputs the standard deviation of Hamming distances between all pairs
of votes.  The intuition is that in a polarized electorate each two
voters either have very similar or close-to-opposite preferences
(factor of $2/|C|$ ensures values in $[0,1]$).
\newcommand{\myfrac}[2]{\frac{#1}{#2}}
\begin{definition}
    For an election $E = (C,V)$, the Pairwise Polarization index is defined as
    \[
    \pairpol(E) = \textstyle 2 \sqrt{\myfrac{\sum_{u, v \in V} (\ham(u,v) - \sum_{x, y \in V}\ham(x,y)/|V|^2)^2}{|V|^2 |C|^2}}.
    \]
\end{definition}

\subsection{Outer Diversity}
Recently, \citet{fal-was-sor-szu:c:outer-diversity} proposed an
alternative diversity index, based on the average distance between the
votes in the election and all possible votes.  If the average distance
is small, the votes are well spread in the whole space, i.e., diverse.
It was originally proposed in the ordinal setting and for domains
(sets of votes without specified multiplicity) rather than elections,
but we adapt it to our setting as follows.

For two collections of votes over the same set of candidates,
$U = (u_1,\dots,u_k)$ and $V = (v_1,\dots,v_n)$,
by $\ham(U,V)$ we mean
the average Hamming distance
between $n$ copies of votes in $U$
and $k$ copies of votes in $V$, matched
so as to minimize this distance, i.e.:
\[
    \textstyle \ham(U,V) = \frac{1}{kn} \min_{\pi : [kn]\to[kn]}\sum_{i=1}^{kn}\ham(u_i,v_{\pi(i)}),
\]
where $\pi$ is a bijection and indices from $U$ are taken modulo $k$
and from $V$ modulo $n$.

For a candidate set $C$, let $\mathcal{U}_C$ consist of a single copy of
each vote from $\{0,1\}^{|C|}$.  Then, in principle, we could measure
diversity of an election $E=(C,V)$ based on $\ham(V, \mathcal{U}_C)$,
i.e., the average distance between votes in $V$ and all possible
votes.  However, if the saturation of $E$ is very small or large,
the value of $\ham(V, \mathcal{U}_C)$ is always substantial (as most votes
in $\mathcal{U}_C$ have moderate saturation).  Thus, no election with
small/large saturation could be considered diverse according to such a
measure.  To circumvent this issue, instead of $\mathcal{U}_C$ we use
$p \hbox{-}\mathcal{U}_C$, defined as a collection in which the number of
copies of each vote is proportional to the probability of drawing such
a vote from $p$-IC
(so, $\nicefrac{1}{2}\hbox{-}\mathcal{U}_C =\mathcal{U}_C$).

We would like the index to be maximum and equal to $1$ if $V = p \hbox{-}\mathcal{U}_C$, and minimum and equal to $0$ if $V$ is a singleton. From the following lemma we get that for a single vote $v$ with saturation $p$, it holds that $\ham(\{v\}, p \hbox{-}\mathcal{U}_C)=2p(1-p)$.

\begin{lemma}
    \label{lem:ham:dist:single-vote:un}
    For candidate set $C$,
    $q \in \{0,\frac{1}{|C|},\dots,1\}$ and $v \in \{0,1\}^{|C|}$ with $q |C|$ ones it holds that:
    $$\ham(\{v\},p\hbox{-}\mathcal{U}_C) = p(1-q) + q(1-p).$$
\end{lemma}

\begin{definition}
\label{def:outdiv}
  For an election $E = (C,V)$ with $\satr(E) = p$,
  the \emph{Outer Diversity} index is defined as
  \[
    \textstyle\outdiv(E) = 1 - \frac{1}{2p(1-p)}\ham(V, p\hbox{-}\mathcal{U}_C).
  \]
\end{definition}

As discussed,
the value $1$ obtained for $V = p\hbox{-}\mathcal{U}_C$,
and a single vote $v$ always yields the value of $0$.
It remains to prove that this is indeed the minimum and maximum value.

\begin{proposition}
  \label{prop:out-div:normalization}
  For each election $E$, $\outdiv(E) \in [0,1].$
\end{proposition}

\subsection{Comparison of the Indices}
Let us now analyze the values of the indices given in
\Cref{tab:index-values}. First, we see that Central Diversity fails
the resampling test of saturation independence, whereas the other
diversity indices pass it.
For polarization indices the test is not meaningful,
as we expect low polarization for all values of $p$ and $\phi$ parameters.
We also see that clustering-based diversity indices give close-to-$1$ values for
Diagonal and IC elections (including Lin-IC). Outer diversity performs
surprisingly poorly on these elections, giving them low values, but we
have verified that its values are much closer to $1$ for analogous
elections with a notably larger number of voters. Still, this
dependence on election size seems a major drawback of the index. This
is quite disappointing, given how well the outer diversity idea worked
for ordinal preference
domains~\citep{fal-was-sor-szu:c:outer-diversity}.  Further, all
diversity indices give higher values for the noisy variants of
respective elections (see the values of $2$-Party versus
$\noisy(2$-Party$,0.6)$ and Triangle versus $\noisy($Triangle$,0.6)$).

All polarization indices give value $1$ for $2$-Party and lower values
for its noisy variant, as well as for $3$- and $4$-Party, as
expected. Pairwise Polarization still gives nearly-$1$  for
$(\nicefrac{1}{3},\nicefrac{1}{3})$-$2$-Party whereas Central and PCC Polarization give
lower ones, but both approaches could be justified.
Overall, Pairwise Polarization tends to give highest values.
For example, for Triangle it gives nearly $\nicefrac{1}{2}$, while the
other indices give values much closer to $0$.
Both approaches could be justified, either by Triangle's similarity
to $2$-Party, or by saying that in Triangle all voters agree on best-to-worst ranking
of the candidates, but have different approval thresholds.

\section{Experiments}\label{sec:experiments}

\begin{figure}[t]
\centering
\begin{minipage}{.5\textwidth}
  \centering
  \includegraphics[width=0.95\linewidth, trim=0 0 0 0, clip]{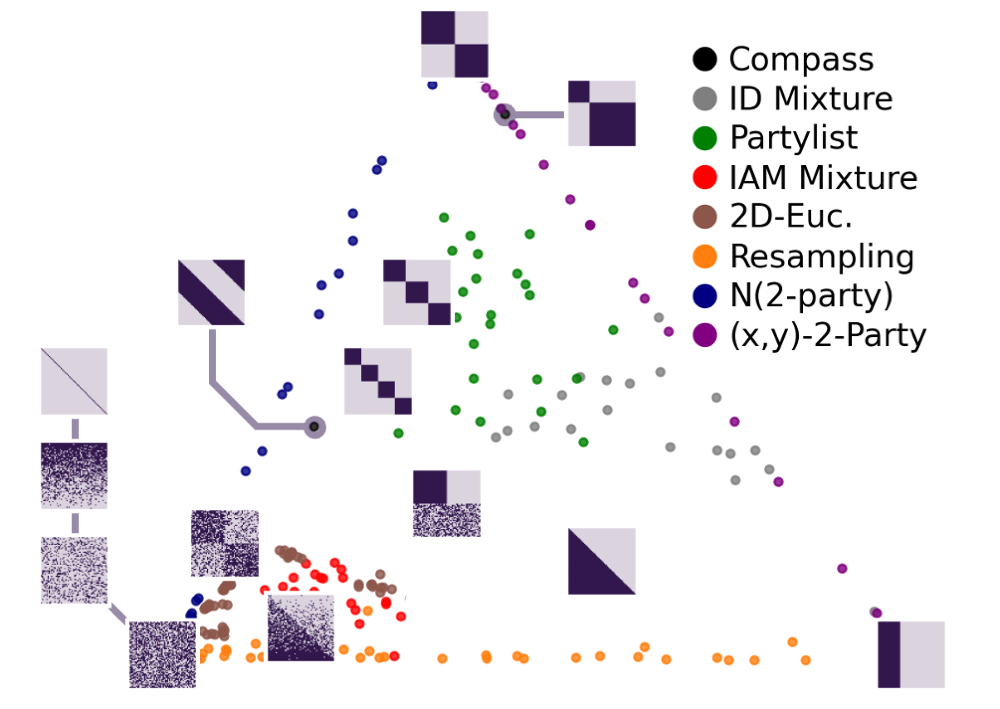}
  \captionof{figure}{Map of synthetic elections.}
  \label{fig:synthetic_map}
\end{minipage}%
\begin{minipage}{.5\textwidth}
  \centering
  \includegraphics[width=0.95\linewidth, trim=0 0 0 0, clip]{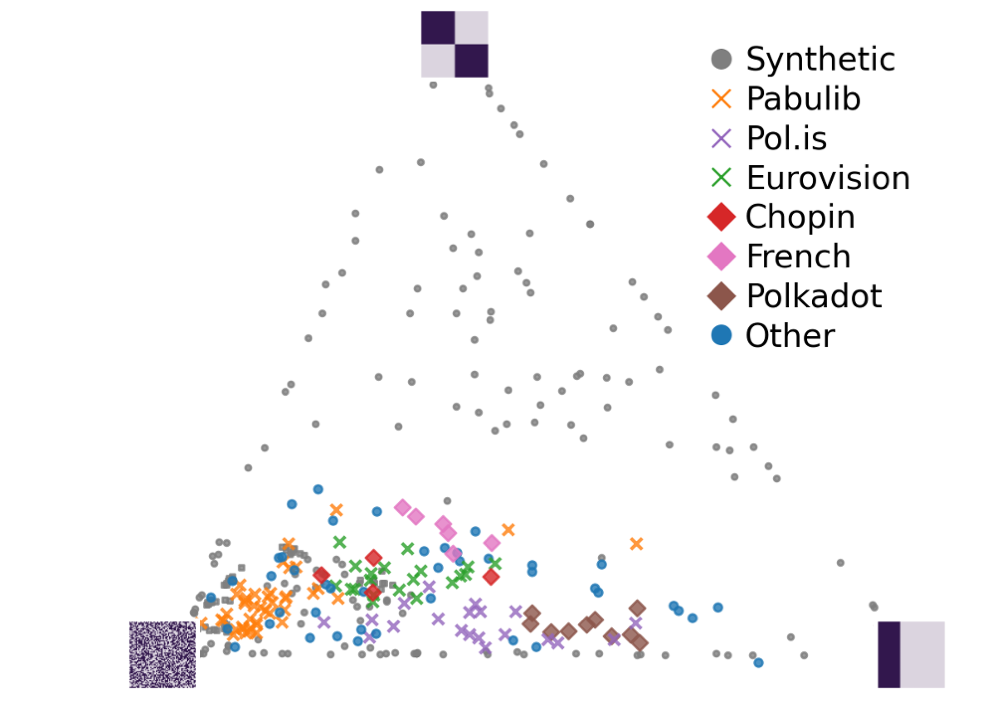}
  \captionof{figure}{Map of real-life elections.}
  \label{fig:real_life_map}
\end{minipage}
\end{figure}

To present experimental results for our indices, we use the \emph{map of elections} framework~\citep{szu-boe-bre-fal-nie-sko-sli-tal:j:map}. We construct the map as follows. First, we assemble a collection of elections. Next, we compute pairwise distances between elections using the feature distance, as in the work of \citet{fal-mer-nun-szu-was:c:diff-sizes}. Specifically, let $f(E)=(f_1(E),\dots,f_k(E))$ denote the feature vector of election $E$, where $f_1,\dots,f_k$ are our indices. For two elections $E$ and $F$, we define their distance as the (standard) Euclidean distance between their feature vectors, $d_f(E,F)=\|f(E)-f(F)\|_2$.
Finally, we embed the resulting distance matrix using multidimensional scaling (MDS~\citep{kru:j:mds}).

For the map, we use the following three indices: $\pccagr$, $\pccdiv$, and $\pccpol$.\footnote{We chose these indices as their sum is the most consistently close to 1 among those studied (see \Cref{app:complementarity} for details), which shows that they are complementary and, hence, together provide a good high-level summary of the nature of a given election. For elections with more than 200 candidates or more than 1000 voters, we computed the indices using sampled subsets of the elections.}
In \Cref{fig:synthetic_map}, we show a map based on synthetic elections, while \Cref{fig:real_life_map} presents a map based on real-life ones (all our datasets will be publicly available). Each point corresponds to a single election. We describe the selected/generated elections in the ~\Cref{app:dataset}.

Interestingly, the synthetic map forms a triangle whose vertices correspond to three extreme cultures: ID (maximizing agreement), IC (maximizing diversity), and 2-party (maximizing polarization). The edge between IC and ID is populated by Resampling elections; the edge between 2-party and IC is filled by $\noisy$(2-Party) elections; and the edge between 2-party and ID is filled by $(x,y)$-$2$-Party elections. Both 2D-Euclidean (brown) and IAM (red) elections lie close to IC, reflecting high diversity together with low polarization and agreement. Finally, the partylist (green) elections occupy the upper region of the triangle, indicating high polarization.

Turning to the real-life map (\Cref{fig:real_life_map}), we observe that most elections exhibit very low polarization and, on average, substantially higher diversity than agreement. Compared with other real-world datasets, the Polkadot blockchain data (stakeholders voting over validators) shows relatively high agreement.
Both Polkadot and Pol.is (users approving comments on an online deliberation platform) exhibit particularly low polarization. Intuitively, this may reflect a shared “search for ground truth” dynamic in both settings.
Notably, the Eurovision (song contest) data looks similar to the Chopin (piano competition) one, which is interesting given that both involve people evaluating musical performances. 
In comparison to other real-life data, the French (presidential) elections seem relatively polarized.
Regarding Pabulib (participatory budgeting voting), with few exceptions, the data shows high diversity, with low agreement and low polarization. (This may stem from the fact that, in most cities, people can approve only a very limited number of projects, so the saturation is extremely low).
To conclude, it is worth noting that our map of real-life elections resembles the one from the work of~\citet{fal-mer-nun-szu-was:c:diff-sizes}, in which (while focusing on ordinal elections), most real-life instances appeared in a similar region of the map, i.e., low polarization, mid to high diversity, and low to mid agreement.

Finally, we note that MDS embeddings of maps in \Cref{fig:synthetic_map,fig:real_life_map}
achieve mean multiplicative distortions of 1.014 and 1.010, respectively,
indicating extremely good representation.
In particular, this is far more accurate than the original
maps presented in \cite[Section 4.3.2]{szu-boe-bre-fal-nie-sko-sli-tal:j:map},
where overall distortion is between 1.21 and 1.26.

\section{Future Work}
We foresee three directions for future work: First, a deeper study of
agreement indices in ordinal elections, going beyond the index we
mention in \Cref{sec:aa} and Kemeny rule. Second, finding more
principled diversity and polarization indices, also for polarization
with more than two disagreeing groups. Third, looking for further
features of (approval) election and designing indices for them (such
as, e.g., measuring if voters approve different candidates
independently).

\section*{Acknowledgments}
This project has received funding from the European Research Council (ERC) under
the European Union’s Horizon 2020 research and innovation programme (grant
agreement No 101002854), from the European Union under the project Robotics and advanced industrial production (reg. no. CZ.02.01.01/00/22\_008/0004590),
and from the Foundation for the University of
Geneva.
T. W\k{a}s was supported by the UK Engineering and Physical Sciences Research Council (EPSRC) under grant EP/X038548/1.
The research presented in this paper has been partially supported by the funds of Polish Ministry of Science and Higher Education assigned to AGH University of Science and Technology.
Further, it was inspired by research proposed for NCN project
UMO-2025/58/A/ST6/00371.
\begin{center}
  \includegraphics[width=3cm]{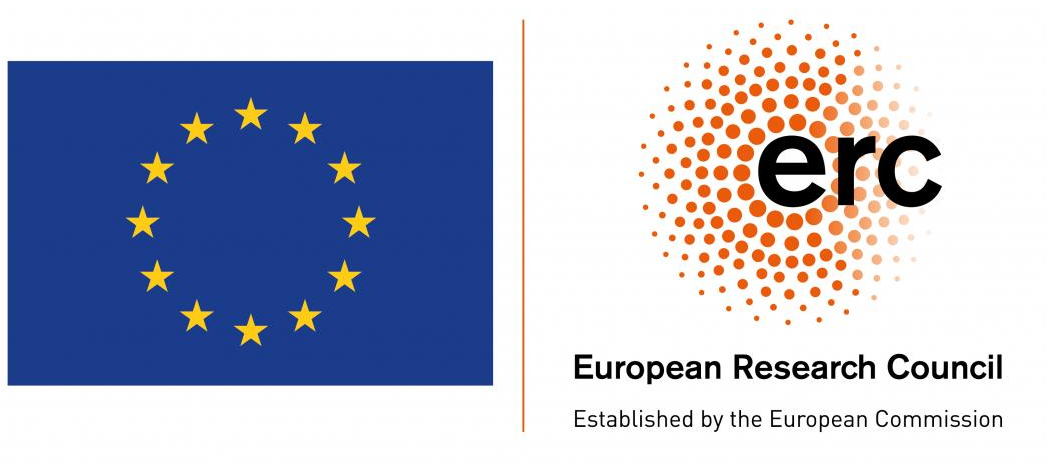}
  \includegraphics[width=3cm]{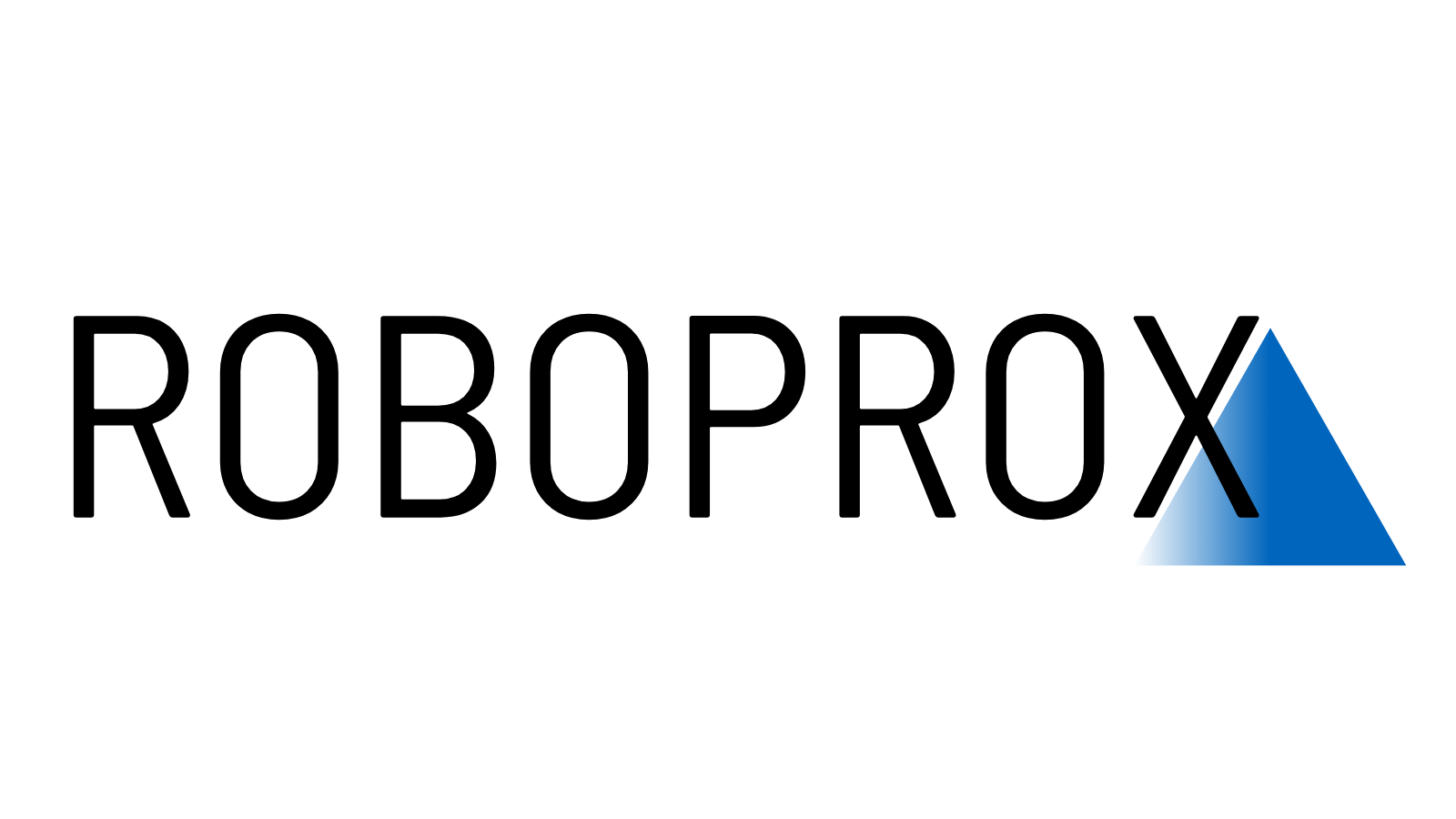}
\end{center}

\bibliography{bib}

\appendix

\clearpage

\section{Additional Tables}

\newcommand{\resamplingentry}[2]{\includegraphics[width=3cm]{img/resampling-sizes/#1\_#2.png}}

\newcommand{\resamplingrow}[2]{
  \raisebox{4.5\height}{#1} &
  \resamplingentry{app\_agr}{#2} &
  \resamplingentry{cntr\_agr}{#2}&
  \resamplingentry{pair\_agr}{#2}&
  \resamplingentry{jacc\_agr}{#2}&
  \resamplingentry{pcc\_agr}{#2}&
  \resamplingentry{pccplus\_agr}{#2}
  \\
}

\newcommand{\resamplingdow}[2]{
  \raisebox{4.5\height}{#1} &
  \resamplingentry{cntr-div}{#2} &
  \resamplingentry{pcc-div}{#2}&
  \resamplingentry{out-div}{#2}&
  \resamplingentry{cntr-pol}{#2}&
  \resamplingentry{pcc-pol}{#2}&
  \resamplingentry{ham-std}{#2}
  \\
}

\begin{table*}
  \scalebox{0.8}{
  \begin{tabular}{c|cccccc}
    \toprule
    size & \makecell{ Approval \\ Agreement} 
         & \makecell{ Central \\ Agreement}
         & \makecell{ Hamming Pairwise \\ Agreement}
         & \makecell{ Jaccard Pairwise \\ Agreement}    
         & \makecell{ PCC \\ Agreement}
         & \makecell{ PCC$^+$ \\ Agreement}\\
    \midrule
    \resamplingrow{$10 \times 10$}{10\_10}
    \resamplingrow{$20 \times 20$}{20\_20}
    \resamplingrow{$60 \times 60$}{60\_60}
    \bottomrule
  \end{tabular}
}

\rule{0cm}{0.3cm}

  \scalebox{0.8}{
  \begin{tabular}{c|cccccc}
    \toprule
    size & \makecell{ Central \\ Diversity} 
         & \makecell{ PCC \\ Diversity}
         & \makecell{ Outer \\ Diversity}
         & \makecell{ Central \\ Polarization}    
         & \makecell{ PCC \\ Polarization}
         & \makecell{ Pairwise \\ Polarization}\\
    \midrule
    \resamplingdow{$10 \times 10$}{10\_10}
    \resamplingdow{$20 \times 20$}{20\_20}
    \resamplingdow{$60 \times 60$}{60\_60}
    \bottomrule
  \end{tabular}

  }
  \caption{\label{tab:agr-sizes}Results of the resampling experiment for elections of different sizes.}
\end{table*}

\Cref{tab:agr-sizes} shows the results of the resampling experiment
for agreement indices and elections of different sizes.

\section{Missing Proofs}

Given two
votes, $u$ and $v$, let $n_{11}, n_{10}, n_{01}$, and
$n_{00}$ be the number of candidates approved in both $u$ and $v$, only
$u$, only $v$, and neither of them.
Then we have:
\begin{align}
  \textstyle
  \pcc(u,v) = \frac{n_{00} \cdot n_{11} - n_{01} \cdot n_{10}}{\sqrt{(n_{10} + n_{11}) \cdot (n_{00} + n_{01}) \cdot (n_{01} + n_{11}) \cdot (n_{00} + n_{10})}}.
  \label{eq:binary-pcc}
\end{align}

\subsection[Proof of Proposition~\ref{pro:av-agr-characterization}]{Proof of \Cref{pro:av-agr-characterization}}
Let $E = (C,V)$ be an election. By definition of the Approval Agreement, we have
that:
\[\textstyle
  \appragr(E) = \frac{1}{|C|}\sum_{c \in C}\left|1-2\frac{|A(c)|}{|V|}\right|.
\]
Hence, the value of $\appragr(E)$ is $1$ if and only if for each
candidate $c \in C$ it holds that
$\left|1-2\frac{|A(c)|}{|V|}\right|=1$, which happens exactly if
$|A(c)|=0$ or $|A(c)|=|V|$. This means that $\appragr(E)=1$ if and
only if $E$ is an identity election. On the other had,
$\appragr(E) = 0$ if and only if for each candidate $c \in C$ we have
$\left|1-2\frac{|A(c)|}{|V|}\right|=0$, which happens exactly when
$2\frac{|A(c)|}{|V|} = 1$, i.e., $|A(c)| = \nicefrac{|V|}{2}$.~\qed

\subsection[Proof of Proposition~\ref{pro:central-characterization}]{Proof of \Cref{pro:central-characterization}}
Let $E = (C,V)$ be an election. By definition of Central Agreement, we have
that:
\[
  \cntragr(E)= 1 - \frac{\sum_{v_i \in V} \ham(v_i, u)}{|V| \cdot \min(\avl(E), \revavl(E))}.
\]
Hence, $\cntragr(E)$ is equal to $1$ exactly if every summand is equal
to $0$, which happens when all the votes are identical. Hence,
$\cntragr(E) = 1$ if and only if $E$ is an identity election.

To characterize when $\cntragr(E) = 0$, let us first introduce an auxiliary
function. By central Hamming distance (CHD) of election $E = (C,V)$ we mean:
\[\textstyle
  \chd(E) = \sum_{v_i \in V} \ham(v_i, u).
\]
Naturally, we have:
\[ \cntragr(E) = 1 - \frac{\chd(E)}{|V| \cdot \min(\avl(E),
    \revavl(E))}.
\]
Let us rephrase $\chd(E)$ as follows: Given an
election $E = (C,V)$ and one of its central votes $u \in \cen(E)$, for
each candidate $c_j$, we define
$\chd(c_j) = \sum_{v_i \in V} |v_i[j]-u[j]|$. Then, by changing the
order of summation, we can have:
\begin{align}
  \chd(E) &\textstyle = \sum_{v_i \in V} \ham(v_i, u) \nonumber \\ 
          &\textstyle = \sum_{v_i \in V}\sum_{c_j \in C}|v_i[j]-u[j]| \label{eq:chd}\\
          &\textstyle =\sum_{c_j \in C}\sum_{v_i \in V} |v_i[j]-u[j]| = \sum_{c_j \in C} \chd(c_j).\nonumber
\end{align}
Using this formulation, we characterize the elections with a given
size and saturation, for which $\chd(E)$ is largest (specifically, it
is easier to verify the claims in the proof of
Lemma~\ref{lem:chd-largest} by using \Cref{eq:chd}).
\begin{lemma}\label{lem:chd-largest}
  The value $\chd(E)$ is largest among elections with $m$ candidates,
  $n$ voters, and saturation $\alpha$ if and only if (a) all
  candidates are approved in some vote in $\cen(E)$, or (b) all
  candidates are disapproved in some vote in $\cen(E)$, or (c) for
  each candidate $c_j$ we have that
  $|A(c_j)| \in \{ \lfloor \nicefrac{n}{2} \rfloor, \lceil
  \nicefrac{n}{2} \rceil\}$.
\end{lemma}
\begin{proof}
  Let us first consider case (b). We know that there is an election
  $E$ with $m$ candidates, $n$ voters, and saturation
  $\satr(E) = \alpha$, such that all candidates are disapproved in
  some central vote in $\cen(E)$.  We arbitrarily partition the
  candidates into those that can receive approvals and those that can
  provide them. Then we perform a sequence of operations where, in
  each of them, we move a single approval from a providing candidate
  to a receiving one. We observe that the $\CHD$ values remains
  constant, until the first operation where for the receiving
  candidate $c_k$ all the central votes for the current election
  approve it. When this happens, the value of $\CHD$ drops by
  one. From this point on, each operation either maintains the current
  value of $\CHD$ or decreases it further. Since, with an appropriate
  sequence of operations and an appropriate choice of providing and
  receiving candidates, we can reach an arbitrary elections with the
  given size and saturation, it must be the case that the starting
  election had the maximum value of $\CHD$. The same reasoning applies
  to case (a), but for moving disapprovals rather than approvals.

  Let us consider case (c), assuming that $n$ is odd (if it were even,
  then cases (a) and (b) would have held as well and we have handled
  them already). Consider an analogous sequence of operations as in
  the preceding paragraph, except that the candidates who are
  originally approved by $\lceil \nicefrac{n}{2} \rceil$ voters (call
  them \emph{big} candidates) cannot give approvals to those, who were
  originally approved by $\lfloor \nicefrac{n}{2} \rfloor$ voters
  (call them \emph{small} candidates; this is w.l.o.g., as if we
  wanted to make such swaps, then we could have simply permuted the
  candidates). Now we observe that whenever an approval is moved from
  a small candidate or to a big candidate, then the $\CHD$ value of
  the election decreases. Yet, all our moves are of this form, so the
  original CHD value must have been maximal.

  Note that every election can be reached by moving approvals from one
  in cases (a), (b), or (c). This completes the proof.
\end{proof}

Next, let us consider some election $E = (C,V)$, such that there is a vote
$u \in \cen(E)$ where all candidates are disapproved. By the above lemma, we
know that its $\chd(E)$ value is largest for its size and saturation. Further,
we note that:
\begin{align*}
  \chd(E) &= \textstyle \sum_{v_i \in V} \ham(v_i, u) \\
          &= \textstyle \sum_{v_i \in V} |A(v_i)| = |V| \cdot \avl(E).
\end{align*}
It is also immediate that $|V| \cdot \avl(E) \leq |V| \cdot \revavl(E)$ in this
case. Similarly, if there is $u \in \cen(E)$ that approves all the
candidates, then we have $\chd(E) = |V| \cdot \revavl(E)$ and
$\revavl(E) \leq \avl(E)$. This means that if either all candidates in
$E$ are approved by at least half of the voters or all candidates in
$E$ are disapproved by at least half of the voters, then
$\cntragr(E) = 0$.

To finish the proof, it remains to consider an election $E = (C,V)$
where some candidates are approved by fewer than half of the voters
(denote the set of these candidates as $C_{<}$), some candidates are
approved by more than half of the voters (denote the set of these
candidates as $C_{>}$), and some (possibly zero) are approved by
exactly half of the voters (denote the set of these candidates as
$C_{=}$). Let us form $E_{<} = (C_{<},V)$, $E_{>} = (C_{>},V)$ and
$E_{=} = (C_{=},V)$. The following holds (note that
$\chd(E_{=}) = |V|\avl(E_{=})=|V|\revavl(E_{=})$):
\begin{align*}
  \chd(E) & = \chd(E_{<}) + \chd(E_{>}) + \chd(E_{=}) \\
          & = |V|(\avl(E_{<}) + \revavl(E_{>}) + \avl(E_{=}))
\end{align*}
One can easily verify that this value is strictly smaller than both
$|V|\avl(E)$ and $|V|\revavl(E)$ (the former inequality holds because
$\avl(E) = \avl(E_{<})+\avl(E_{>})+\avl(E_{=})$ while
$\revavl(E_{>}) < \avl(E_{<})$; the latter inequality holds for
analogous reasons).  Consequently, $\cntragr(E) > 0$.~\qed

\subsection[Proof of Proposition~\ref{pro:cntr-value}]{Proof of \Cref{pro:cntr-value}}

We use the notation as in the proof of
\Cref{pro:central-characterization}, in
Equation~\eqref{eq:chd}. Consequently, we have election $E = (C,V)$
and $\chd(E) = \sum_{c_j \in C}\chd(c_j)$, where:
\[
  \chd(c_j) = \sum_{v_i \in V}|v_i[j] - u[j]|
\]
for some $u \in \cen(E)$. We note that:
\[
  \chd(c_j) = \begin{cases}
    |A(c_j)| & \text{if } u[j] = 0,\\
    |V|-|A(c_j)| & \text{if } u[j]=1.
  \end{cases}
\]
Further, we have that $u[j] = 0$ whenever $|A(c_j)|\leq |V|/2$ and
$u[j]=1$ otherwise (strictly speaking, if $|A(c_j)|$ were equal to
$|V|/2$ then we could also take $u[j]=1$ in this case, but this
would not change the logic of the argument). Consequently, we have
that:
\[
  \chd(c_j) = \textstyle \frac{|V|}{2} - \left|\frac{|V|}{2} - |A(c_j)|\right|.
\]
After substituting this to the formula for $\cntragr(E)$, we get:
\begin{align*}
  \cntragr(E) &= \textstyle  1 - \frac{\chd(E)}{|V|\cdot \min(\avl(E),\revavl(E))} \\
              &= \textstyle  1 - \frac{\sum_{c_j \in C}\left(\frac{|V|}{2} - \left|\frac{|V|}{2} - |A(c_j)|\right| \right)}{|V| \cdot \min(\avl(E),\revavl(E))}\\
              &= \textstyle  1 - \frac{\frac{|V|}{2}\sum_{c_j \in C}\left( 1 - \left|1 - \frac{2|A(c_j)|}{|V|}\right| \right)}{|V| \cdot \min(\avl(E),\revavl(E))}\\
              &= \textstyle  1 - \frac{\sum_{c_j \in C} \left(1 - \left|1 - \frac{2|A(c_j)|}{|V|}\right| \right)}{2 \cdot \min(\avl(E),\revavl(E))}
\end{align*}
This completes the proof.~\qed

\subsection[Proof of Proposition~\ref{pro:pcc-is-index}]{Proof of \Cref{pro:pcc-is-index}}
This result is, in essence, folk knowledge. We provide an argument for
the sake of completeness. Let $x$ and $y$ be two vectors of dimension
$m$. We form two new vectors, $x'$ and $y'$, such that for each $i \in [m]$
we have:
\begin{align*}
  x'[i] = \textstyle \frac{x[i]-\overline{x}[i]}{\sqrt{\sum_{j=1}^m(x[i]-\overline{x}[i])^2}} \text{ and }
  y'[i] = \textstyle \frac{y[i]-\overline{y}[i]}{\sqrt{\sum_{j=1}^m(y[i]-\overline{y}[i])^2}}.
\end{align*}
We observe that $\pcc(x,y)$ is equal to the dot product of $x'$ and
$y'$.

Now let us consider some election $E = (C,V)$, where
$V = (v_1, \ldots, v_n)$. We form an $n \times n$ matrix $M$, where
entry in the $i$-th row and $j$-th column is $\pcc(v_i,v_j)$. By the
argument from the previous paragraph, we know that $M$ is a Gram
matrix (i.e., a matrix of dot products of a sequence of
vectors). Since all Gram matrices are semi-definite positive, and sums
of entries of all semi-definite positive matrices are nonnegative, we
have that:
\[
  \sum_{v_i \in V}\sum_{v_j \in V}\pcc(v_i,v_j) \geq 0.
\]
On the other hand, since for each two votes $v_i$ and $v_j$ we have that $\pcc(v_i,v_j) \leq 1$,
it is immediate that:
\[
  \sum_{v_i \in V}\sum_{v_j \in V}\pcc(v_i,v_j) \leq n^2.
\]
Together, this means that $\pccagr(E) \in [0,1]$.~\qed

\subsection[Proof of Theorem~\ref{thm:pairwise-value}]{Proof of \Cref{thm:pairwise-value}}
Let $E = (C,V)$ be an election with candidate set
$C = \{c_1, \ldots, c_m\}$.
We define its expected Hamming distance
(EHD) as follows:
\[
  \textstyle \ehd(E) = \sum_{u \in V}\sum_{v \in V} \ham(u,v).
\]
Naturally, we have:
\[
  \pairagr(E) = \textstyle \nicefrac{\ehd(E)}{2|V|^2|C| \cdot \satr(E)(1-\satr(E))}.
\]
By definition of $\ehd(E)$ and changing the order of summation, we
have:
\begin{align*}
  \ehd(E) &= \textstyle  \sum_{u \in V}\sum_{v \in V}\ham(u,v)\\
          &= \textstyle  \sum_{u \in V}\sum_{v \in V}\sum_{c_j \in C}|u[j] - v[j]|\\
          &= \textstyle  \sum_{c_j \in C}\sum_{u \in V}\sum_{v \in V}|u[j] - v[j]|.
\end{align*}
Now we note that for each candidate $c_j$, there are $|A(c_j)|$ voters
who approve it and $n-|A(c_j)|$ who do not.  Hence, for a given
candidate $c_j \in C$, we have that:
\begin{align*}
  \textstyle \sum_{u \in V}&\textstyle\sum_{v \in V}|u[j] - v[j]|\\
                             &=\textstyle |A(c_j)| \cdot \sum_{v \in V}|1 - v[j]| \\
                             &+\textstyle (n-|A(c_j)|) \cdot \sum_{v \in V}|0-v[j]|\\
                             &=|A(c_j)|(n-|A(c_j)|) + (n-|A(c_j)|)|A(c_j)|\\
                             &= 2|A(c_j)|(n-|A(c_j)|).
\end{align*}
By substituting this equality into the definition of $\ehd(E)$ we obtain:
\[\textstyle
  \ehd(E) = \sum_{c_j \in C}2|A(c_j)|(n-|A(c_j)|),
\]
which substituted to the formula for $\pairagr(E)$ gives:
\[
  \pairagr(E) = \textstyle \frac{\sum_{c_j \in C}|A(c_j)|(n-|A(c_j)|)}{|V|^2 |C| \cdot \satr(E)(1-\satr(E))}.
\]
This completes the proof.~\qed

\subsection[Proof of Theorem~\ref{thm:pairwise-characterization}]{Proof of \Cref{thm:pairwise-characterization}}
It is easy to verify that, indeed, each of the four indices takes
value $1$ if and only if the input election is an identity
one. Further, closer inspection of the definition of Jaccard Pairwise
Agreement indicates that it assumes value $0$ if and only if all
voters approve disjoint sets of candidates (otherwise one of the
summands in its definition would be greater than $0$). It remains to
consider when Hamming Pairwise Agreement assumes value $0$.

We use the same notation as in the proof of \Cref{thm:pairwise-value}
and we first prove the next lemma, which characterizes for which
elections $E$ of a given size and saturation value $\ehd(E)$ is
largest.

\begin{lemma}\label{lem:ehd-largest}
  Let $m$ and $n$ be two positive integers, and let $\alpha \in [0,1]$
  be a rational number such that $\alpha \cdot n m$ is an integer.
  Let $E = (C,V)$ be an election with $m$ candidates, $n$ voters, and
  $\satr(E) = \alpha$. Value $\ehd(E)$ is largest among elections of
  its size and saturation if for each two candidates $c_j, c_k \in C$,
  $\big||A(c_j)|-|A(c_k)|\big|\leq 1$.
\end{lemma}
\begin{proof}
  Let $E = (C,V)$ be an arbitrary election such that
  $C = \{c_1, \ldots, c_m\}$, $V = (v_1, \ldots, v_n)$, and
  $\satr(E) = \alpha$. We claim that if there are two candidates
  $c_j, c_k \in C$ such that $c_k$ has at least two more approvals
  than $c_j$ then by swapping two approvals within $E$ we can obtain a
  new election $E'$ such that $\ehd(E') > \ehd(E)$.

  Let us first express $\ehd(E)$ in a convenient way. Let us define
  function $f(x) = x(1-x)$ and for each candidate $c_t \in C$, let
  $p_t$ be its approval probability, i.e.,
  $p_t = \nicefrac{|A(c_t)|}{n}$. Now, using the formula for $\ehd(E)$
  given in the proof of \Cref{thm:pairwise-value}, we have:
  \begin{align*}
    \ehd(E) &\textstyle = \frac{2}{n^2}\sum_{c_t \in C} n\cdot f(p_t)\\
            &\textstyle = \frac{2}{n}\sum_{c_t \in C} f(p_t)\\
            &\textstyle = \frac{2}{n}\left( f(p_j) + f(p_k) \right) + X(E),
  \end{align*}
  where
  $X(E) = \frac{2}{n}\sum_{c_t \in C \setminus\{c_j,c_k\}}
  f(p_t)$. Next, we form election $E'$ that is identical to $E$
  except that we remove a single approval from $c_k$ and add a single
  approval to $c_j$. We have that:
  \begin{align*}
    \textstyle
     \ehd(E')= \frac{2}{n}\left( f(p_j+\nicefrac{1}{n}) + f(p_k-\nicefrac{1}{n}) \right) + X(E),
  \end{align*}
  To establish that $\ehd(E') > \ehd(E)$, we need to show that
  \begin{equation}
    f(p_j+\nicefrac{1}{n}) - f(p_j) > f(p_k) - f(p_k - \nicefrac{1}{n}).
    \label{eq:ehd}
  \end{equation}
  If we have that both $p_j, p_k \leq 0.5$ or both $p_j,p_k \geq 0.5$
  then it is clear that this holds (as the derivative of $f(x)$ is
  strictly decreasing and assumes value $0$ for $x=0.5$). On the other
  hand, if $p_j < 0.5$ and $p_k > 0.5$ then either the left-hand side
  of \eqref{eq:ehd} is positive and the right-hand side is at most
  zero, or the left-hand side if zero and the right-hand size is
  negative (this follows from the shape of function $f(x)$, which
  grows for $x \in [0,0.5]$ and then symmetrically decreases for
  $x \in [0.5,1]$, and the fact that $c_k$ has at least two more
  approvals than $c_j$ in $E$).

  By repeating the above reasoning, we find that elections where the
  numbers of approvals received by some two candidates differ by more
  than one cannot have the largest $\EHD$ value. By
  \Cref{thm:pairwise-value} we know that an election's $\EHD$ value---for a
  given number of voters---depends only on the candidates' approval
  scores. Consequently, an election has the highest $\EHD$ value among
  those with a given numbers of candidates and voters, and a given
  saturation, if and only if approval scores of its candidates differ
  at most by one.
\end{proof}

Now let us consider an election $E = (C,V)$, where each candidate has
the same approval score equal to $T$. In such an election, we have
$\satr(E) = T/|V|$. By \Cref{thm:pairwise-value} and substituting
candidates' approval scores into the formula, we have that:
\begin{align*}
  \pairagr(E) &=\textstyle  1 - \frac{|C|T(|V|-T)}{|V|^2 |C| \cdot \satr(E)(1-\satr(E))}\\
  &=\textstyle  1 - \frac{|C|T(|V|-T)}{|V|^2 |C| \cdot T/|V|(1-T/|V|)} \\
  &=\textstyle  1 - \frac{|C|T(|V|-T)}{T(|V|-T)} = 0.
\end{align*}
Next, we observe that for an election where candidates' approval
scores differ by at most one (and, indeed, there are two candidates
with different scores) it holds that $\pairagr(E) > 0$. To this end,
let us consider election $E = (C,V)$ with $m$ candidates and $n$
voters, where $T$ voters approve all the candidates, $1$ voter
approves $X < m$ candidates, and remaining $n-T-1$ voters do not
approve any candidates. We have:
\begin{align*}
  \ehd&(E) = T(0 + (m-X) + m(n-T-1)) \\
      &+ T(m-X) + 0 + X(n-T-1) \\
      &+ (n-T-1)(mT + X + 0) \\
      &= 2T(m-X) + 2mT(n-T-1) + 2X(n-T-1)\\
      &= 2mT -2TX + 2nmT -2mT^2 \\
      &-2mT +2Xn -2TX -2X \\
      &=2nmT -4TX - 2mT^2 + 2Xn -2X\\
      &=\textstyle n^2m(\frac{2T}{n} - \frac{4TX}{n^2m} -\frac{2T^2}{n^2} + \frac{2X}{nm} - \frac{2X}{n^2m}).
\end{align*}
For election $E$ we have that $\satr(E) = \frac{mT+X}{nm}$, so:
\begin{align*} 
  \pairagr(E) &=\textstyle 1 - \frac{n^2m(\frac{2T}{n} - \frac{4TX}{n^2m} -\frac{2T^2}{n^2} + \frac{2X}{nm} - \frac{2X}{n^2m})}{2n^2m(\frac{mT+X}{nm})(1-\frac{mT+X}{nm})} \\
              &=\textstyle 1 - \frac{\frac{2T}{n} - \frac{4TX}{n^2m} -\frac{2T^2}{n^2} + \frac{2X}{nm} - \frac{2X}{n^2m}}{2(\frac{mT+X}{nm})(1-\frac{mT+X}{nm})}\\
              &=\textstyle 1 - \frac{\frac{2T}{n} - \frac{4TX}{n^2m} -\frac{2T^2}{n^2} + \frac{2X}{nm} - \frac{2X}{n^2m}}{\frac{2T}{n}-\frac{4TX}{n^2m}-\frac{2T^2}{n^2}+\frac{2X}{nm}-\frac{2X^2}{n^2m^2}} > 0.
\end{align*}
The final inequality follows by computing the numerator and
denominator. This concludes the proof as, by \Cref{lem:ehd-largest},
for every other election of this size and saturation, $\pairagr(E)$ is even larger.

\subsection[Proof of Proposition~\ref{pro:pcc}]{Proof of \Cref{pro:pcc}}
Let $E = (C,V)$ be a $k$-Party election. We write $m$ and $n$ to
denote $|C|$ and $|V|$, and we assume that both $m$ and $n$ are
divisible by $k$. Given two votes, $u$ and $v$ in $V$, that approve
the same groups of $\nicefrac{m}{k}$ candidates, we have that:
\[
  \pcc(u,v) = 1.
\]
On the other hand, if $u$ and $v$ approve disjoint sets of $\nicefrac{m}{k}$
candidates each, then, by Equation~\eqref{eq:binary-pcc} we have:
\begin{align*}
  \pcc(u,v) = \textstyle \frac{-\frac{m}{k}\cdot \frac{m}{k}}{\sqrt{\frac{m}{k}\cdot(m-\frac{m}{k})\cdot\frac{m}{k}\cdot(m-\frac{m}{k})}} 
   = \textstyle -\frac{\frac{m^2}{k^2}}{\frac{m^2}{k}(1-\frac{1}{k})} = -\frac{1}{k}\frac{1}{1-\frac{1}{k}}.
\end{align*}
Then, we have that (the second equality follows by substituting the
above values for PCCs):
\begin{align*}
  \pccagr(E) &=\textstyle \frac{1}{n^2}\sum_{v_i \in V}\sum_{v_j \in V} \pcc(v_i,v_j)\\
             &=\textstyle \frac{1}{n^2}\cdot n \cdot( \frac{n}{k} - (n-\frac{n}{k})\frac{1}{k}\frac{1}{1-\frac{1}{k}})\\
             &=\textstyle (\frac{1}{k} - \frac{1}{k}) = 0.
\end{align*}
Similarly, we have:
\begin{align*}
  \pospccagr(E) &=\textstyle \frac{1}{n^2}\sum_{v_i \in V}\sum_{v_j \in V} \max(0,\pcc(v_i,v_j))\\
             &=\textstyle \frac{1}{n^2}\cdot n \cdot \frac{n}{k} = \frac{1}{k}.
\end{align*}
This completes the proof.~\qed

\subsection[Proof of Proposition~\ref{prop:pairagr-pccagr-sometimes-equal}]{Proof of \Cref{prop:pairagr-pccagr-sometimes-equal}}

Let $E=(C,V)$ be an election such that $0 < |A(u)| = |A(v)| < 1$ for all $u,v \in V$.

First, we derive an explicit dependence between the $\pcc(u,v)$ and $\ham(u,v)$ for any pair of $u,v \in V$.
We note that an implicit dependence has been known in the literature on association coefficients (see, e.g., \cite{war:j:association-coeff});
for completeness we provide the exact formula in our notation.

Specifically, \citet{war:j:association-coeff} considers Yule's $\phi$ coefficient
(denoted there by $S_{\textrm{Yule1}}$) \cite{yul:j:phi-coefficient},
which coincides with PCC for binary vectors,
and also the simple matching coefficient $S_{\textrm{SM}}$~\cite{sok-mic:j:association-btw-attributes},
which in our notation equals to the number of matches, i.e.,
$S_{\textrm{SM}}(u,v) = m-\ham(u,v)$.
\citet{war:j:association-coeff} shows that,
under fixed marginals,
$S_{\textrm{Yule1}}$ can be expressed as an affine transformation of $S_{\textrm{SM}}$,
but do not provide an explicit formula,
which is as follows.

\begin{proposition}\label{prop:pcc-ham-relation}
  If $0 < |A(u)| = |A(v)| = pm < m$ holds then we have
  $\pcc(u,v) = 1 - \frac{\ham(u,v)}{2m(1-p)p}$.
\end{proposition}
\begin{proof}
  Let $n_{11}, n_{10}, n_{01}, n_{00}$ denote the numbers of candidates approved in both $u$ and $v$,
  only in $u$, only in $v$, and in neither of them, respectively.
  By these definitions we have
  \begin{align}
    n_{11}+n_{10} = n_{11}+n_{01} = pm.\label{eq:nxy-marginal}
  \end{align}
  This implies:
  \begin{align*}
    n_{00} &= m - n_{11} -n_{10} -n_{01}\\
           &\stackrel{\eqref{eq:nxy-marginal}}{=} m - n_{11} -(pm -n_{11}) -(pm -n_{11})\\
           &= m -2pm + n_{11}.
  \end{align*}

  Now consider the numerator of $\pcc(u,v)$ (defined in \Cref{eq:binary-pcc}).
  \begin{align*}
      &n_{11} \cdot n_{00} - n_{01} \cdot n_{10}\\
    \stackrel{\eqref{eq:nxy-marginal}}{=}~&n_{11} \cdot (m -2pm + n_{11})  - (pm - n_{11})^2\\
    =~&n_{11} \cdot m -(pm)^2.
  \end{align*}

  The denominator of $\pcc(u,v)$ equals to:
  \begin{align*}
      & \sqrt{(n_{10} + n_{11}) \cdot (n_{00} + n_{01}) \cdot (n_{01} + n_{11}) \cdot (n_{00} + n_{10})}\\
      & = \sqrt{ pm \cdot (m-pm) \cdot pm \cdot (m-pm)}\\
      & = pm \cdot (m-pm) = m^2 \cdot p(1-p).
  \end{align*}

  Putting both together we obtain:
  \begin{align}
    \pcc(u,v) &= \frac{n_{11} \cdot m -(mp)^2}{m^2 \cdot p(1-p)}\nonumber\\
              &= \frac{n_{11} -mp^2}{m p(1-p)}.\label{eq:pcc-as-n11-function}
  \end{align}

  By the definition of the Hamming distance we have
  $
    \ham(u,v) = n_{10} + n_{01} = 2(pm-n_{11})
  $
  which is equivalent to
  \begin{align}
    n_{11} = pm -\frac{\ham(u,v)}{2}.\label{eq:n11-as-ham}
  \end{align}

  Putting these together we obtain the final equality:
  \begin{align*}
    \pcc(u,v) &\stackrel{\eqref{eq:pcc-as-n11-function},\eqref{eq:n11-as-ham}}{=} \frac{mp -\frac{\ham(u,v)}{2} -mp^2}{m p(1-p)}\\
    &\stackrel{\phantom{(6),(7)}}{=} \frac{mp-mp^2}{m p(1-p)} - \frac{\frac{\ham(u,v)}{2}}{m p(1-p)}\\
    &\stackrel{\phantom{(6),(7)}}{=} 1 - \frac{\ham(u,v)}{2m p(1-p)}.
  \end{align*}
\end{proof}

Since $|A(v)| = pm$ for all $v \in V$,
we have $\avl(E) = pm$,
and thus
$\satr(E) = \frac{\avl(E)}{m} = p$
and together with \Cref{prop:pcc-ham-relation} we can prove the final equality:
\begin{align*}
\pccagr(E) \stackrel{\textrm{Def.}~\ref{def:pair-agree-indices}}{=} \hspace{2pt}&\frac{1}{n^2}\sum_{u \in V}\sum_{v \in V}\pcc(u,v)\\
            \stackrel{\textrm{Prop.}~\ref{prop:pcc-ham-relation}}{=}
               \hspace{0pt}&\frac{1}{n^2}\sum_{u \in V}\sum_{v \in V} \left( 1 - \frac{\ham(u,v)}{2m p(1-p)} \right)\\
             = \hspace{11pt}&1 -\sum_{u \in V}\sum_{v \in V} \frac{\ham(u,v)}{2 n^2 m p(1-p)}\\
             \stackrel{\textrm{Def.}~\ref{def:pair-agree-indices}}{=} \hspace{2pt}&\pairagr(E).
\end{align*}
This completes the proof.~\qed

\subsection[Proof of Lemma~\ref{lem:ham:dist:single-vote:un}]{Proof of \Cref{lem:ham:dist:single-vote:un}}
\begin{proof}
    We will prove that both sides of the equation
    can be linked to the ratios
    of the sizes of two combinatorial sets
    defined as follows.
    Note that in $p\hbox{-}\mathcal{U}_C$
    every vote appears $a^r \cdot b^{|C|-r}$ times,
    where $a$ and $b$ are some integers such that $p=a/(a+b)$ and $r$ is the number of ones in the vote.
    Let $F$ be a set of functions from the set of $m$ numbers to the set of $a+b$ numbers, $f : [m] \rightarrow [a+b]$.
    Then, let $X$ be set of pairs $(x,f) \in [m] \times F$,
    such that $f(x) \in [a]$, if $x \in [(1-q)m]$,
    or $f(x) \in \{a+1,\dots,a+b\}$,
    if $x \in \{(1-q)m+1,\dots,m\}$.
    We claim that
    (1) $\ham(\{v\},p\hbox{-}\mathcal{U}_C) = |X|/(m\cdot|F|)$ and
    (2) $|X|/(m\cdot|F|) = p(1-q) + q(1-p).$

    (1)
    First, observe that
    \begin{equation}
    \label{eq:pUN:size}
        |F| = 
        (a+b)^m = 
        \sum_{k=0}^m \binom{m}{k} a^kb^{m-k} = 
        |p\hbox{-}\mathcal{U}_C|.
    \end{equation}
    
    Now, without loss of generality, let us assume
    that the first $(1-q)m$ positions of $v$ are $0$s
    and the last $qm$ positions are $1$s.
    According to the definition of $p\hbox{-}\mathcal{U}_C$,
    each vote $u \in \{0,1\}^m$ appears there
    $a^\alpha b^\beta$ times,
    where $\alpha$ is the number of ones in $u$
    and $\beta$ is the number of zeros.
    This corresponds to the number of functions $f$
    that give values in $[a]$ to exactly those arguments in $[m]$
    that belong to $A(u)$.
    Now, for each such function $f$,
    we can choose $x$ in the number of different ways
    that is equal to the number of indices in $A(u)$
    that are also in $[(1-q)m]$ plus the number of indices
    that are not in $A(u)$ but are in $\{(1-q)m+1,\dots,m\}$,
    i.e., the symmetric difference between $A(u)$ and $\{(1-q)m+1,\dots,m\}$.
    But this is exactly the Hamming distance between $u$ and $v$.
    Hence, $|X|$ is equal to the sum of the (unnormalized) Hamming distances between $v$ and each vote from $p\hbox{-}\mathcal{U}_C$.
    Combining this with \cref{eq:pUN:size}, we get that indeed
    $\ham(\{v\},p\hbox{-}\mathcal{U}_C) = |X|/(m\cdot|F|)$.

    (2)
    To express the size of $X$ in a different way,
    we first fix $x \in [m]$, and then select matching $f$.
    We split $X$ into two disjoint subsets depending on the value of $x$, in a way that $X_1 = \{(x,f) \in X : x \in [(1-q)m]\}$ and $X_2 = \{(x,f) \in X : x \in \{(1-q)m+1,\dots,m\}\}$.

    In $X_1$ we can select $x$ in $(1-q)m$ different ways,
    then $f(x)$ can have $a$ different values,
    and $f(y)$ for each $y \in [m] \setminus \{x\}$
    can have $(a+b)$ different values.
    Thus, $|X_1| = (1-q)m \cdot a \cdot (a+b)^{m-1}$.
    Analogously, $|X_2| = q \cdot m \cdot b \cdot (a+b)^{m-1}$.
    Since $X_1$ and $X_2$ are disjoint and
    $X = X_1 \cup X_2$, we get that
    $|X| = |X_1| + |X_2| =
    ((1-q)a + qb)m(a+b)^{m-1}$.
    Thus, by \cref{eq:pUN:size}
    \[
        \frac{|X|}{m|F|} = \frac{(1-q)a + qb}{a+b} = (1-q)p +q(1-p).
    \]
\end{proof}

\subsection[Proof of Proposition~\ref{prop:out-div:normalization}]{Proof of \Cref{prop:out-div:normalization}}
\begin{proof}
    Clearly, for every election $E = (C,V)$
    and $p = \satr(E)$,
    it holds that
    \(
        \ham(V, p\hbox{-}\mathcal{U}_C) \ge 0,
    \)
    hence $\outdiv_h^{\satr}(E) \le 1$.
    For the lower bound,
    by a series of (in)equalities, we will show that
    \(
        \ham(V, p\hbox{-}\mathcal{U}_C) \le 2(p - p^2),
    \)
    which will imply that $\outdiv_h^{\satr}(E) \ge 0$.

    Denote $V = \{v_1,\dots,v_n\}$.
    Let us start by considering copying $n$ times
    each vote in $p\hbox{-}\mathcal{U}_C$,
    we denote the resulting collection of votes as $n \cdot p\hbox{-}\mathcal{U}_C$.
    Observe that introducing these additional copies
    does not affect the Hamming distance from $V$
    (as otherwise it would mean that there is a fractional matching
    giving less total weight than an integral matching,
    which is not possible~\cite{lov-plu:b:matchings}).
    Hence,
    \[
        \ham(V,p\hbox{-}\mathcal{U}_C) = \ham(V, n \cdot p\hbox{-}\mathcal{U}_C).
    \]

    To compute $\ham(V, n \cdot p\hbox{-}\mathcal{U}_C)$ we match
    the votes from $V$ to those in $n \cdot p\hbox{-}\mathcal{U}_C$ in an optimal way.
    Another way of matching the votes would be to match each separate set of copies of votes in $p\hbox{-}\mathcal{U}_C$ to a single, each time different vote in $V$.
    Thus, we get
    \begin{align*}
        \ham(V, p\hbox{-}\mathcal{U}_C) &=
        \ham(V, n \cdot p\hbox{-}\mathcal{U}_C) \\ &\le
        \frac{1}{n} \sum_{i=1}^n \ham(\{v_i\}, p\hbox{-}\mathcal{U}_C).
    \end{align*}

    Now, let us denote a fraction of ones in $v_i$ by $q_i = |A(v_i)|/|C|$.
    From \Cref{lem:ham:dist:single-vote:un}, we get that
    \begin{align*}
        \ham(V, p\hbox{-}\mathcal{U}_C) &\le
        \frac{1}{n} \sum_{i=1}^n \ham(\{v_i\}, p\hbox{-}\mathcal{U}_C) \\ &=
        \frac{1}{n} \sum_{i=1}^n
            \big(p \cdot (1-q_i) + (1-p) \cdot q_i\big) \\ &=
        p \left(1 - \frac{1}{n} \sum_{i=1}^n q_i \right) +
        (1-p) \cdot \frac{1}{n} \sum_{i=1}^n q_i \\ &=
        2 p (1 - p),
    \end{align*}
    where the last equality we can obtain from the fact that
    $\sum_{i=1}^n q_i = \satr(E) = p$.
    This concludes the proof.
\end{proof}

\section{Further Data Analysis}
In this appendix, we elaborate on the measure of complementarity of computed indices, and provide correlation matrix of their values.

\subsection{Complementarity}
\label{app:complementarity}

\citet{fal-kac-sor-szu-was:c:div-agr-pol-map} suggest that the elections they consider approximately fit into a simplex in the space of their agreement, diversity, and polarization indices.
We would like to measure how much this is the case for all triples of our indices.
To this end, we use the following \emph{measure of complementarity}
(we note that similar measures have been used to measure a diversification of a portfolio in financial sciences, see e.g.,~\cite{cho-coi:j:diversification,han-lin-wan:j:diversification}).

\begin{definition}
    Let $\mathbf{x}=(x_1,\dots,x_n)$,
    $\mathbf{y}=(y_1,\dots,y_n)$, and
    $\mathbf{z}=(z_1,\dots,z_n)$,
    be three vectors of index values.
    We define their \emph{complementarity} as
    one minus the standard deviation
    of the sum of the indices
    divided by the sum of the standard deviation of the indices, i.e.,
    \[
        \cmpl(\mathbf{x},\mathbf{y},\mathbf{z}) = 1 - \frac{\std(x_1 + y_1 + z_1, \dots, x_n + y_n + z_n)}{\std(\mathbf{x}) + \std(\mathbf{y}) + \std(\mathbf{z})},
    \]
    where
    \[\textstyle
        \std(a_1,\dots,a_n) = \sqrt{\frac{1}{n}\sum_{i=1}^n(a_i - \frac{1}{n}\sum_{j=1}^na_j)^2}.
    \]
\end{definition}

Intuitively, complementarity measures how close is the sum of indices to being constant as compared to the variability of the indices alone.
If the sum is constant, complementarity would give value equal to 1.
On the other hand, for any vector $\mathbf{x} \in \mathbb{R}^n$,
we have $\cmpl(\mathbf{x},\mathbf{x},\mathbf{x})=0$.

We checked the values of the complementarity measure for all triples of our indices.
The four that gave the highest values where:
\begin{itemize}
\item $\pccagr$, $\pccdiv$, $\pccpol$: the value of $0.9016$,
\item $\cntragr$, $\cntrdiv$, $\pccpol$: the value of $0.8793$,
\item $\cntragr$, $\cntrdiv$, $\cntrpol$: the value of $0.8775$, and
\item $\pairagr$, $\cntrdiv$, $\pccpol$: the value of $0.8748$.
\end{itemize}
This motivated us to focus on the
$\pccagr$, $\pccdiv$, $\pccpol$ map in \Cref{sec:experiments}.

For comparison,
in the original paper by \citet{fal-kac-sor-szu-was:c:div-agr-pol-map} that introduced diversity-agreement-polarization framework, the complementarity of their indices for the ordinal elections in their \emph{extended dataset}
is equal to $0.8818$
(we have normalized the values of the indices to have the maximum value of $1$, as it is not the case there).
In the follow-up paper by
\citet{fal-mer-nun-szu-was:c:diff-sizes},
the complementarity of the indices of the election in the main dataset is equal to $0.9086$.

\begin{table*}[t]
\centering
\setlength{\tabcolsep}{0.2pt}
\small
\begin{tabular}{lccccccccccccc}
\toprule
& \multicolumn{1}{p{0.8cm}}{\rotatebox{53}{$\satr$}}
& \multicolumn{1}{p{0.8cm}}{\rotatebox{53}{$\appragr$}} 
& \multicolumn{1}{p{0.8cm}}{\rotatebox{53}{$\cntragr$}} 
& \multicolumn{1}{p{0.8cm}}{\rotatebox{53}{$\pairagr$}} 
& \multicolumn{1}{p{0.8cm}}{\rotatebox{53}{$\pccagr$}} 
& \multicolumn{1}{p{0.8cm}}{\rotatebox{53}{$\pospccagr$}} 
& \multicolumn{1}{p{0.8cm}}{\rotatebox{53}{$\jaccpairagr$}} 
& \multicolumn{1}{p{0.8cm}}{\rotatebox{53}{$\cntrdiv$}} 
& \multicolumn{1}{p{0.8cm}}{\rotatebox{53}{$\pccdiv$}} 
& \multicolumn{1}{p{0.8cm}}{\rotatebox{53}{$\outdiv$}} 
& \multicolumn{1}{p{0.8cm}}{\rotatebox{53}{$\cntrpol$}} 
& \multicolumn{1}{p{0.8cm}}{\rotatebox{53}{$\pccpol$}} 
& \multicolumn{1}{p{0.8cm}}{\rotatebox{53}{$\pairpol$}} \\
\midrule
satr & -- & \qqncor{0.496} & \qqpcor{0.425} & \qqpcor{0.163} & \qqpcor{0.032} & \qqpcor{0.115} & \qqpcor{0.644} & \qqncor{0.345} & \qqncor{0.121} & \qqncor{0.111} & \qqpcor{0.132} & \qqpcor{0.149} & \qqpcor{0.294} \\
app\_agr & \qqncor{0.461} & -- & \qqncor{0.143} & \qqpcor{0.177} & \qqpcor{0.320} & \qqpcor{0.100} & \qqncor{0.202} & \qqpcor{0.211} & \qqncor{0.057} & \qqpcor{0.049} & \qqncor{0.247} & \qqncor{0.302} & \qqncor{0.490} \\
cntr\_agr & \qqpcor{0.374} & \qqpcor{0.095} & -- & \qqpcor{0.554} & \qqpcor{0.399} & \qqpcor{0.350} & \qqpcor{0.580} & \qqncor{0.567} & \qqncor{0.328} & \qqncor{0.251} & \qqpcor{0.094} & \qqpcor{0.100} & \qqpcor{0.196} \\
pair\_agr & \qqpcor{0.243} & \qqpcor{0.328} & \qqpcor{0.909} & -- & \qqpcor{0.709} & \qqpcor{0.453} & \qqpcor{0.397} & \qqncor{0.437} & \qqncor{0.413} & \qqncor{0.157} & \qqpcor{0.080} & \qqpcor{0.042} & \qqpcor{0.044} \\
pcc\_agr & \qqpcor{0.106} & \qqpcor{0.480} & \qqpcor{0.690} & \qqpcor{0.816} & -- & \qqpcor{0.650} & \qqpcor{0.281} & \qqncor{0.288} & \qqncor{0.531} & \qqncor{0.163} & 0.009 & \qqncor{0.052} & -0.007 \\
pcc+\_agr & \qqpcor{0.173} & \qqpcor{0.242} & \qqpcor{0.613} & \qqpcor{0.691} & \qqpcor{0.855} & -- & \qqpcor{0.422} & \qqncor{0.555} & \qqncor{0.827} & \qqncor{0.421} & \qqpcor{0.279} & \qqpcor{0.270} & \qqpcor{0.302} \\
jacc\_agr & \qqpcor{0.812} & \qqncor{0.127} & \qqpcor{0.689} & \qqpcor{0.656} & \qqpcor{0.516} & \qqpcor{0.603} & -- & \qqncor{0.592} & \qqncor{0.421} & \qqncor{0.318} & \qqpcor{0.219} & \qqpcor{0.195} & \qqpcor{0.287} \\
cntr\_div & \qqncor{0.392} & \qqpcor{0.198} & \qqncor{0.668} & \qqncor{0.626} & \qqncor{0.420} & \qqncor{0.699} & \qqncor{0.707} & -- & \qqpcor{0.636} & \qqpcor{0.469} & \qqncor{0.477} & \qqncor{0.490} & \qqncor{0.464} \\
pcc\_div & \qqncor{0.174} & \qqncor{0.115} & \qqncor{0.495} & \qqncor{0.576} & \qqncor{0.667} & \qqncor{0.913} & \qqncor{0.577} & \qqpcor{0.832} & -- & \qqpcor{0.498} & \qqncor{0.366} & \qqncor{0.407} & \qqncor{0.363} \\
out\_div & \qqncor{0.134} & \qqpcor{0.040} & \qqncor{0.374} & \qqncor{0.371} & \qqncor{0.295} & \qqncor{0.503} & \qqncor{0.398} & \qqpcor{0.642} & \qqpcor{0.609} & -- & \qqncor{0.222} & \qqncor{0.252} & \qqncor{0.280} \\
cntr\_pol & \qqpcor{0.181} & \qqncor{0.379} & \qqncor{0.059} & \qqncor{0.059} & \qqncor{0.110} & \qqpcor{0.353} & \qqpcor{0.284} & \qqncor{0.649} & \qqncor{0.567} & \qqncor{0.461} & -- & \qqpcor{0.613} & \qqpcor{0.497} \\
pcc\_pol & \qqpcor{0.180} & \qqncor{0.399} & \qqncor{0.028} & \qqncor{0.081} & \qqncor{0.151} & \qqpcor{0.354} & \qqpcor{0.276} & \qqncor{0.669} & \qqncor{0.604} & \qqncor{0.504} & \qqpcor{0.929} & -- & \qqpcor{0.619} \\
pair\_pol & \qqpcor{0.257} & \qqncor{0.516} & \qqpcor{0.030} & \qqncor{0.081} & \qqncor{0.125} & \qqpcor{0.360} & \qqpcor{0.295} & \qqncor{0.640} & \qqncor{0.569} & \qqncor{0.444} & \qqpcor{0.817} & \qqpcor{0.912} & -- \\
\bottomrule
\end{tabular}
\caption{Correlations between our features based on the values for elections in the main dataset. The values in the top-right triangle are Kendall's $\tau$ correlation coefficients, while those in the bottom-left triangle are Pearson's $\rho$ correlation coefficients.}
\label{tab:correlations}
\end{table*}

\subsection{Correlations}

In \Cref{tab:correlations} we report the Pearson's $\rho$ and Kendall's $\tau$ correlation coefficients for each pair of our indices.

We observe that agreement indices
(except for Approval Agreement),
diversity indices,
and polarization indices all form relatively well correlated clusters which confirms our intuitions that they are measuring more or less the same features.
Across these clusters there is usually either negative correlation or a lack of correlation.
The only exception is the fact that
Jaccard and PCC$^+$ Pairwise Agreement indices
are slightly correlated with the polarization indices.
This is due to the fact that
Jaccard and PCC$^+$ Pairwise Agreement indices capture the local agreement, which means that the presence of voter groups that have similar preference inside each group,
but possibly very dissimilar across the groups,
increases the value of these two agreement indices.
However, the same structures also increase the polarization indices.

Finally, we note that the Saturation index is not significantly correlated with any other index, except for Jaccard Pairwise Agreement.
This can be explained by noting that these are the only two indices that we consider that do not treat approvals and disapprovals symmetrically.

\section{Detailed Datasets Description}
\label{app:dataset}
In this appendix we provide a complete description of the data used to produce \Cref{fig:synthetic_map,fig:real_life_map}.
We begin by providing details on exact distributions used to generate synthetic elections as well as sources of real-world data and the methods of its preprocessing.
We conclude with \Cref{tab:all_data},
in which for each election in the dataset
we provide the values of all considered features
as well as its icon akin to those in \Cref{tab:index-values}.

\subsection{Synthetic Data}
Let us start by describing elections generated by us using synthetic distributions.

\paragraph{Compass}
Compass elections include exactly those 
that are listed in \Cref{tab:index-values},
one copy each.

\paragraph{Impartial Culture}
We generated ten IC elections with 100 candidates and 1000 voters, two for each value of $p = \{0.1,0.3,0.5,0.7,0.9\}$.

\paragraph{Lin-IC}
We generated five Lin-IC elections with 100 candidates and 1000 voters each, as described in \Cref{sec:election_data}.

\paragraph{Resampling}
We generated 50 resampling elections drawing $\phi$ and $p$ uniformly at random from $[0,1]$,
each with 100 candidates and 1000 voters.

\paragraph{N(2-Party)}
We generated 25 noisy 2-party elections, each with 100 voters and 100 candidates. We sampled $\phi$ uniformly at random from $[0,1]$.

\paragraph{($\boldsymbol{x}$,$\boldsymbol{y}$)-2-Party}
We generated 25 ($x$,$y$)-2-Party elections, with 100 voters and 100 candidates each.
For each we chose $x$ and $y$ uniformly at random and independently from $[0,1]$.

\paragraph{Party-List}
We generated 25 uneven party-list elections, each with 100 voters and 100 candidates, in the following way.
First we sampled the number of parties $k$
from a Poisson distribution with $\lambda=1$
and added $3$, so that we have at least $3$ parties.
Next, we uniformly at random chose two partitions of 100 into $k$ numbers.
One of the partitions was then used as the number of voters in each party, while the other was used as the number of candidates.

\paragraph{ID Mixture}
For $p \in [0,1]$ and integer $k$,
in $k,p$-ID Mixture model,
we first split the voters randomly into $k$ groups.
Then, each group has identical approvals
drawn from $p$-IC.
We generated 20 ID Mixture elections,
each with 100 candidates and 1000 voters,
one for each combination of $k \in \{2, 3, 4, 5\}$
and $p \in \{0.1, 0.3, 0.5, 0.7, 0.9\}$.

\paragraph{IAM Mixture}
For a given $k$, we generate an IAM Mixture election by first splitting the voters randomly into $k$ groups,
then for each group and candidate $c$
drawing $p_c$ uniformly at random from $[0,1]$,
and finally sampling a sub-election for each group
using the Independent Approval Model
described in \Cref{sec:election_data}.
We generated 20 IAM Mixture elections, each with 100 candidates and 1000 voters,
five for each $k \in \{2,3,4,5\}$.

\paragraph{2D-Euclidean}
We sampled 50 election from the 2D Euclidean model,
replicating first five distributions from \cite{god-bat-sko-fal:c:2d} with 10 elections in each (each with 100 candidates and 1000 voters). Specifically, we assumed that both voters and candidates are sampled uniformly from $[0,1]\times [0,1]$ square, and then:
\begin{itemize}
    \item each voter $i$ approves all candidates in radius $r=0.117$ from $i$'s position,
    \item each voter $i$ approves all candidates in radius $r=0.167$ from $i$'s position,
    \item each voter $i$ approves all candidates in radius $r_i$ from $i$'s position, where $r_i$ is drawn uniformly from $[0,0.5]$,
    \item each voter $i$ approves $10$ closest candidates, or
    \item each voter $i$ approves some number of closest candidates, that is drawn uniformly from $\{1,2,\dots,100\}$.
\end{itemize}

\subsection{Real-World Elections}
Let us now cover elections that are based on real-world data. We include elections from Preflib~\cite{mat-wal:c:preflib}, Pabulib~\cite{fal-fli-pet-pie-sko-sto-szu-tal:c:pabulib}, as well as a number of datasets that we have gathered or converted ourselves.

\paragraph{Chopin}
These elections are based on the XIX International Chopin Piano Competition. Each election is based on one of four stages of the competition. In each, every member of the jury gives a score from 1 to 25 to every participant of a given stage (except for a relatively rare case in which the participant is a student of the member of the jury). The jury is fixed and consists of 17 members, which are voters in our election, and the number of participants, i.e., candidates, decreases with each stage (starting from 84 in the first stage and ending at 11 in the final stage) as only the best candidates pass to the next stage. For each stage, we converted a score above 20 to an approval.

\paragraph{French Elections}
Each of these datasets consists of exit-poll surveys conducted during the French presidential elections of 2007 \cite{Baujard2025:french-elections-2007}, 2012 \cite{Baujard2025:french-elections-2012}, 2017 \cite{bou-bla-bau-dur-ige-lan-lar-las-leb-mer:m:french-elections-2017}, and 2022 \cite{Delemazure2024-tc-FrenchElections2022}.

\paragraph{Jester}
The Jester approval elections were derived from a Kaggle dataset%
\footnote{\url{https://www.kaggle.com/datasets/aakaashjois/jester-collaborative-filtering-dataset}} in which users assigned numerical scores to jokes~\cite{Goldberg2001-pc-Jester}.
In order to convert the data into approval profile, jokes were ranked according to their scores, with unranked jokes placed randomly at the bottom of the ranking. A fixed top percentage of jokes was then labeled as approved candidates, while the remaining jokes were considered unapproved.

\paragraph{MovieLens}
This dataset is a subset of MovieLens 20M Dataset that can be found on Kaggle~\cite{har-kon:j:movielens}.%
\footnote{\url{https://www.kaggle.com/datasets/grouplens/movielens-20m-dataset}}
The subset consists of 50 films (candidates) and 423 users (voters) each of which rates each film with a score from $\{1.0, 1.5, 2.0, \dots, 5.0\}$.
We have converted this subset into three elections by setting an approval for every score that is equal or above $4.0$, $4.5$, or $5.0$.

\paragraph{Networks}
We have included seven elections that have been converted from canonical network datasets.
In each of them we have treated the set of nodes as both voters and candidates with an approval corresponding to an edge from a  voter to a candidate.
Additionally, each node approves itself.

Three networks have been taken from the Mark Newman's collection of network datasets.\footnote{\url{https://websites.umich.edu/~mejn/netdata/}}
These include American College Football Network~\cite{GirNew-2002-CommunityDetection},
Dolphin Social Network~\cite{LusSchBoiHaaSloEtal-2003-DolphinNetwork}, and
Political Blogs Network~\cite{AdaGla-2005-PolBlogs}.
The next four network has been taken from the predefined datasets in the Editor4Centralities online tool.\footnote{\url{https://centrality.mimuw.edu.pl/editor/}}
These include September 11 Attack Terrorist Network~\cite{Kre-2002-September11Network},
Wood-Processing Facility Strike Network~\cite{Mic-1997-StrikeNetwork},
Zachary's Karate Club Network~\cite{Zac-1977-KarateNetwork}, and
Les Miserables Character Network~\cite{Knu-1993-LesMiserablesNetwork}.

\paragraph{Polkadot}
We sampled 10 elections from the Polkadot dataset on Preflib.
Each election represents votes of the stakeholders in the Polkadot blockchain over the set of validators in the Nominated Proof-of-Stake (NPoS) protocol
\cite{boe-bri-cev-geh-san-sch:c:polkadot}.

\paragraph{Our Surveys}
We collected these datasets in an online survey filled predominantly by students. Voters were asked to mark their like or dislike for each of the candidates which then translated to approvals and disapprovals. No restrictions regarding the vote length were imposed.
The topics of the surveys include
colors, fruits, pets, pizza types, and school subjects (each topic corresponds to one election).

\paragraph{Bidding Data}
We include six elections based on bids of reviewers over the papers which they want to review in computer science conferences.
All of the data comes from Preflib.
Three elections comes from AAMAS 2015, 2016, and 2021 editions
and three come from unspecified conferences.
In each case, we have interpreted "yes" bids as an approval, and all other responses (and lack thereof) as a lack of approval.

\paragraph{Campsongs}
These four elections are surveys in which the participants had to state their preferences over CCM-genre music pieces. The data is available on Preflib.

\paragraph{CTU}
This is an approval election over the time of tutorial at the Czech Technical University in Prague.

\paragraph{Eurovision}
We include a sample of 20 election based on the Eurovision Song Contest from the dataset available on Preflib~\cite{boe:phd}. Originally the data consists of rankings over most preferred participants as seen by each member of the jury (the number of ranked participants is fixed each time). We have converted these top-truncated ordinal elections to approval elections, by setting an approval each for each ranked participant.

\paragraph{Posters}
These are two approval elections from the San Sebastian Poster Competition available at Preflib. Two elections were held over two groups of posters (A group) and (B group).

\paragraph{Sushi}
These are two elections based on two surveys on the favorite kinds of sushi~\cite{kam:c:sushi}. In one of them the participants had to indicate most preferred 10 types from the total of 100 types and provide their strict ranking. We have converted this data into an approval election by setting an approval for each of the selected 10 types. In the other survey the participants were allowed to provide a weak ordering of the 10 chosen types (with ties). Here, we included only the kinds that were tied at the first place.

\paragraph{Pol.is}
The data comes from online deliberation platform Pol.is, in which the users can express their support, neutrality, or disapprovement over the comments posted by other users. The data is incomplete, as not every user has to react to every comment. We have transformed all datasets of this type available at Preflib into approval elections by recognizing the support as an approval.

\paragraph{Pabulib}
Finally, we include a sample of participatory budgeting elections from Pabulib~\cite{fal-fli-pet-pie-sko-sto-szu-tal:c:pabulib}.
In each such election,
citizens of a given municipality or district
might express their preferences over the projects that can be implemented in their neighborhood.
Most of the elections that we included are natively approval elections
(the only exception is Toulouse 2019 election
in which the voters could distribute 7 points between the projects---we treated each positive amount of points as an approval).
However, the projects are associated with costs that could impact the voters decisions, and which are ignored for our purposes.
We include five elections from each of the following five cities: Amsterdam, Budapest, Łódź, Toulouse, Warsaw.
The elections included are the ones that have the highest ``quality'' score on the Pabulib webpage.
Apart from that, we included all elections from the districts of Warsaw in 2026.

\subsection{Index Values}
The following is the table with values of our indices for all elections in our main dataset.
For each election we present an icon representing its approval matrix, akin to those presented in \Cref{tab:index-values}.
If the election had more than $500$ candidates (or voters),
we have sampled $500$ candidates (or voters)
for the purpose of drawing the icon.
For most of the elections we ordered the rows (and columns) in the icon based on the Jaccard distances between the votes (and candidates).
Based on the distances we computed the tree of hierarchical clustering and ordered the rows (and columns)
according to the ordering of the leaves in the the tree~\cite{eis-spe-bro-bot:j:clustering-ordering}.

{
\setlength{\tabcolsep}{1.5pt}
\footnotesize
\begin{longtable}{c cccc|cccccc|ccc|ccc}
\toprule
Election & Icon & $|C|$ & $|V|$ & \multicolumn{1}{p{0.8cm}}{\rotatebox{53}{$\satr$}} & \multicolumn{1}{p{0.8cm}}{\rotatebox{53}{$\appragr$}} & \multicolumn{1}{p{0.8cm}}{\rotatebox{53}{$\cntragr$}} & \multicolumn{1}{p{0.8cm}}{\rotatebox{53}{$\pairagr$}} & \multicolumn{1}{p{0.8cm}}{\rotatebox{53}{$\pccagr$}} & \multicolumn{1}{p{0.8cm}}{\rotatebox{53}{$\pospccagr$}} & \multicolumn{1}{p{0.8cm}}{\rotatebox{53}{$\jaccpairagr$}} & \multicolumn{1}{p{0.8cm}}{\rotatebox{53}{$\cntrdiv$}} & \multicolumn{1}{p{0.8cm}}{\rotatebox{53}{$\pccdiv$}} & \multicolumn{1}{p{0.8cm}}{\rotatebox{53}{$\outdiv$}} & \multicolumn{1}{p{0.8cm}}{\rotatebox{53}{$\cntrpol$}} & \multicolumn{1}{p{0.8cm}}{\rotatebox{53}{$\pccpol$}} & \multicolumn{1}{p{0.8cm}}{\rotatebox{53}{$\pairpol$}} \\
\midrule
\endfirsthead
\toprule
Election & Icon & $|C|$ & $|V|$ & \multicolumn{1}{p{0.8cm}}{\rotatebox{53}{$\satr$}} & \multicolumn{1}{p{0.8cm}}{\rotatebox{53}{$\appragr$}} & \multicolumn{1}{p{0.8cm}}{\rotatebox{53}{$\cntragr$}} & \multicolumn{1}{p{0.8cm}}{\rotatebox{53}{$\pairagr$}} & \multicolumn{1}{p{0.8cm}}{\rotatebox{53}{$\pccagr$}} & \multicolumn{1}{p{0.8cm}}{\rotatebox{53}{$\pospccagr$}} & \multicolumn{1}{p{0.8cm}}{\rotatebox{53}{$\jaccpairagr$}} & \multicolumn{1}{p{0.8cm}}{\rotatebox{53}{$\cntrdiv$}} & \multicolumn{1}{p{0.8cm}}{\rotatebox{53}{$\pccdiv$}} & \multicolumn{1}{p{0.8cm}}{\rotatebox{53}{$\outdiv$}} & \multicolumn{1}{p{0.8cm}}{\rotatebox{53}{$\cntrpol$}} & \multicolumn{1}{p{0.8cm}}{\rotatebox{53}{$\pccpol$}} & \multicolumn{1}{p{0.8cm}}{\rotatebox{53}{$\pairpol$}} \\
\midrule
\endhead
\raisebox{3mm}{\makecell{$\nicefrac{1}{3}$-ID}} & \includegraphics[width=0.8cm]{img/longtable/id33\_org.png} & \raisebox{3mm}{\makecell{$60$}} & \raisebox{3mm}{\makecell{$60$}} & \ltbase{0.33} & \ltagr{1.00} & \ltagr{1.00} & \ltagr{1.00} & \ltagr{1.00} & \ltagr{1.00} & \ltagr{1.00} & \ltdiv{0.00} & \ltdiv{0.00} & \ltdiv{0.00} & \ltpol{0.00} & \ltpol{0.00} & \ltpol{0.00} \\
\raisebox{3mm}{\makecell{2-Party}} & \includegraphics[width=0.8cm]{img/longtable/party2\_org.png} & \raisebox{3mm}{\makecell{$60$}} & \raisebox{3mm}{\makecell{$60$}} & \ltbase{0.50} & \ltagr{0.00} & \ltagr{0.00} & \ltagr{0.00} & \ltagr{0.00} & \ltagr{0.50} & \ltagr{0.50} & \ltdiv{0.20} & \ltdiv{0.23} & \ltdiv{0.10} & \ltpol{1.00} & \ltpol{1.00} & \ltpol{1.00} \\
\raisebox{3mm}{\makecell{N(2-Party, 0.6)}} & \includegraphics[width=0.8cm]{img/longtable/noisy206\_org.png} & \raisebox{3mm}{\makecell{$60$}} & \raisebox{3mm}{\makecell{$60$}} & \ltbase{0.49} & \ltagr{0.08} & \ltagr{0.05} & \ltagr{0.01} & \ltagr{0.01} & \ltagr{0.10} & \ltagr{0.33} & \ltdiv{0.66} & \ltdiv{0.82} & \ltdiv{0.29} & \ltpol{0.32} & \ltpol{0.18} & \ltpol{0.24} \\
\raisebox{3mm}{\makecell{3-Party}} & \includegraphics[width=0.8cm]{img/longtable/party3\_org.png} & \raisebox{3mm}{\makecell{$60$}} & \raisebox{3mm}{\makecell{$60$}} & \ltbase{0.33} & \ltagr{0.33} & \ltagr{0.00} & \ltagr{0.00} & \ltagr{0.00} & \ltagr{0.33} & \ltagr{0.33} & \ltdiv{0.33} & \ltdiv{0.30} & \ltdiv{0.13} & \ltpol{0.33} & \ltpol{0.50} & \ltpol{0.63} \\
\raisebox{3mm}{\makecell{4-Party}} & \includegraphics[width=0.8cm]{img/longtable/party4\_org.png} & \raisebox{3mm}{\makecell{$60$}} & \raisebox{3mm}{\makecell{$60$}} & \ltbase{0.25} & \ltagr{0.50} & \ltagr{0.00} & \ltagr{0.00} & \ltagr{0.00} & \ltagr{0.25} & \ltagr{0.25} & \ltdiv{0.45} & \ltdiv{0.40} & \ltdiv{0.15} & \ltpol{0.25} & \ltpol{0.33} & \ltpol{0.43} \\
\end{longtable}
}
\end{document}